\documentclass[journal]{IEEEtran}
\usepackage{graphics} 
\usepackage{epsfig} 
\usepackage{times} 
\usepackage{amsmath} 
\usepackage{amsthm}
\usepackage{amssymb}
\usepackage[utf8]{inputenc}
\usepackage{tikz}\usepackage{enumitem}
\usetikzlibrary{positioning, arrows.meta, automata}
\usetikzlibrary{graphs}
\usepackage[colorlinks=false, urlcolor=blue, pdfborder={0 0 0}]{hyperref}
\usepackage{enumerate}
\usepackage{color}
\usepackage[linesnumbered,boxed,ruled,commentsnumbered]{algorithm2e}
\usepackage{mathrsfs}

\usepackage{subfigure}
\usepackage{comment}
\usepackage{cite}
\usepackage{flushend} 
\usepackage{tcolorbox}

\newtheorem{mydef}{Definition}
\newtheorem{mythm}{Theorem}
\newtheorem{myprob}{Problem}
\newtheorem{mylem}{Lemma}

\newtheorem{mypro}{Proposition}
\newtheorem{myexm}{Example}

\def \O{{\mathcal{O}}}
\def \S{{\mathfrak{S}}}
\def \B{{\mathfrak{B}}}
\def \L{{\mathcal{L}}}
\def \V{{\mathbb{V}}}
\def \UR{{\text{UR}}}
\def \OR{{\text{OR}}}
\def \state{{\textsf{state}}}

\def \vecc{{\textsf{vec}}}

\title{On Prediction-Based Properties \\ of Discrete-Event Systems:\\ Notions, Applications and Supervisor Synthesis}

\author{{Bohan Cui}, \IEEEmembership{Student Member, IEEE}, {Yu Chen}, \IEEEmembership{Student Member, IEEE}, Alessandro Giua, \IEEEmembership{Fellow, IEEE},\\and Xiang Yin, \IEEEmembership{Member, IEEE}

\thanks{This work was supported by  the National Natural Science Foundation of China (62173226, 62061136004, 92367203). }
\thanks{B. Cui, Y. Chen and X. Yin are with the School of Automation and Intelligent Sensing, Shanghai Jiao Tong University, Shanghai 200240, China, and also with the Key Laboratory of System Control and Information Processing, the Ministry of Education of China, Shanghai 200240, China. {\tt  E-mail: \{bohan\_cui, yuchen26, yinxiang\}@sjtu.edu.cn}.
Alessandro Giua is with the Department of Electrical and Electronic Engineering, University of Cagliari, Cagliari 09123, Italy. {\tt  E-mail: giua@unica.it}.}
\thanks{Corresponding Author: Xiang Yin.}}

\IncMargin{0.8em}
\begin{document}

\maketitle

\begin{abstract}
In this work, we investigate the problem of synthesizing property-enforcing supervisors for partially-observed discrete-event systems (DES). Unlike most existing approaches, where the enforced property depends solely on the executed behavior of the system,  here we consider a more challenging scenario in which the property relies on predicted future behaviors that have not yet occurred. 
This problem arises naturally in applications involving future information, such as active prediction or intention protection.
To formalize the problem, we introduce the notion of  \emph{prediction-based properties}, a new class of observational properties tied to the system's future information. 
We demonstrate that this notion is very generic and can model various practical properties, including predictability in fault prognosis and pre-opacity in intention security.
We then present an effective approach for synthesizing supervisors that enforce prediction-based properties.
Our method relies on a novel information structure that addresses the fundamental challenge arising from the dependency between current predictions and the control policy. The key idea is to first borrow information from future instants and then ensure information consistency. This reduces the supervisor synthesis problem to a safety game in the information space.
We prove that the proposed algorithm is both sound and complete, and the resulting supervisor is maximally permissive.
\end{abstract}


\IEEEpeerreviewmaketitle

\section{Introduction}

\subsection{Motivations}

\IEEEPARstart{T}{his} paper investigates the problem of supervisory control for partially observed discrete-event systems (DES). 
A supervisor is a mechanism that regulates the behavior of the system by dynamically enabling or disabling events based on its observation sequence \cite{cassandras2008introduction,wonham2019supervisory}. 
In particular, due to limitations in actuators and sensors, a partial-observation supervisor must account for both uncontrollability and observability issues, ensuring that the closed-loop behavior always satisfies some desired high-level specification despite disturbances and information uncertainty. 
Since the seminal work of Ramadge and Wonham, supervisory control theory has been extensively developed and widely adopted as a formal controller synthesis framework for various engineering systems.
Recent successful applications of supervisory control theory include, for example, production systems \cite{thuijsman2024supervisory}, industrial control systems \cite{reijnen2022supervisory}, multi-robot systems \cite{rosa2024modular} and  lock-bridge systems \cite{reijnen2020modeling}. 

In the partial observation setting, due to the presence of unobservable events, 
the system state cannot be perfectly known, and state estimation needs to be performed based on the observed event sequence, i.e., the information-flow. 
In such cases, formal specifications in system verification and synthesis often pertain to the information-flow generated by the system \cite{hadjicostis2020estimation}. 
A typical requirement involves the current-state estimate, for instance, ensuring that a user gains sufficient information to resolve current ambiguity between different states, 
as in diagnosability \cite{sampath1998active}, observability \cite{lin1988observability}, and detectability \cite{shu2007detectability}. 
Conversely, it may also require that an external intruder lacks sufficient information to uncover critical secrets, as in current-state opacity \cite{lin2011opacity,ritsuka2025joint}.
Furthermore, the system may use observations up to the current instant to infer past states. 
Properties related to such backward inference  include delayed detectability \cite{shu2012delayed} and K-/infinite-step opacity 
\cite{balun2021comparing,wintenberg2022general}.

In systems theory, alongside filtering (current state estimation) and smoothing (delayed state estimation), \emph{prediction} serves as another fundamental component. 
Here, the focus lies in forecasting the system’s future behavior and verifying whether it satisfies certain desired specifications. 
Observational properties related to future information have also been explored in the partially-observed DES literature. Below are two typical application scenarios:
\begin{itemize}
    \item 
    \textbf{Fault Prediction: } 
    To ensure safe operation, system users may need to predict whether the system will enter critical states or exhibit unsafe behavior patterns. 
    If such behavior can be anticipated, proactive actions can be taken to prevent catastrophic outcomes. 
    In \cite{genc2009predictability}, the notions of predictability and its variants were proposed as necessary and sufficient conditions for fault prediction with no missed or false alarms.
    \item 
    \textbf{Intention Security: } 
    Conversely, from a security perspective, some systems may aim to remain as unpredictable as possible. 
    For instance, if a robot intends to visit a critical region, it should prevent malicious observers from inferring its target too far in advance.
    Along these lines, the concept of pre-opacity has been studied in the literature \cite{yang2022secure} to capture the system’s ability to conceal critical intention for future actions.
\end{itemize}

\textbf{Future-Dependency Challenge in Synthesis. }
Most existing works on prediction-related  properties focus solely on the verification problem. 
When the system's future behavior fails to meet desired requirements, such as untimely fault prediction or premature intention disclosure,  it becomes necessary to design supervisors that restrict the system’s behavior, ensuring the closed-loop system satisfies the specifications. 
However, the supervisor synthesis problem for prediction-related properties is significantly more challenging than both its verification counterpart and standard supervisor synthesis for current-state properties (e.g., current-state opacity). The fundamental difficulty comes from the future-dependency issue. Specifically, \emph{evaluating whether the  future behavior of the closed-loop system satisfies desired specifications requires knowledge of future control decisions that have not yet been synthesized.}
This issue notably does not appear in current-state property synthesis, where the separation principle allows control decisions to be evaluated based solely on completed system trajectories. Similarly, the verification problem avoids this complication entirely since it examines fixed system behaviors without any control intervention.

\subsection{Our Results and Contributions}
In this work, we present a general framework for synthesizing property-enforcing supervisors for partially-observed DES with respect to predicted future behaviors. 
Specifically, our main contributions are summarized as follows.
\begin{itemize} 
    \item 
    First, we introduce a new class of observational properties related to the future information of the system, called \emph{prediction-based properties}. 
    To formalize this, for each observation, we consider the set of predicted reachable states for each future instant. 
    We use an evaluation function as a predicate that evaluates whether or not the membership status of each reachable set satisfies some requirement.
    A system is said to satisfy a prediction-based property if this evaluation prediction holds for all possible observations.
    \item 
    We then demonstrate that this new definition of prediction-based properties is very generic and by using specific evaluation functions, it can model many useful properties in different application scenarios, including fault prognosis and intention security. 
    Specifically, we show that notions such as predictability and pre-opacity (with guaranteed performance bounds) can be captured by this framework. 
    Furthermore, we extend the definition to accommodate more complex requirements, such as anonymity for future intentions.
    \item 
    Finally, we present an effective approach for solving the supervisor synthesis problem to enforce prediction-based properties. 
    Our method addresses the fundamental challenge of synthesizing control strategies for future information by introducing a novel technique. 
    The key idea is to first borrow information from future instants, where control decisions have not yet been determined, and then commit to the borrowed future states when making control decisions.  
    Based on this concept, we develop a game-based synthesis algorithm, which is both sound and complete, over a new  information state space that effectively captures all possible future configurations.  
\end{itemize}

\subsection{Related Works}
In the following, we discuss existing works closely related to our work and highlight the key differences with our results.

\subsubsection{ Supervisor Control under Partial Observation}   
Supervisory control under partial observation has been extensively studied since the early development of SCT \cite{lin1988observability, cieslak1988supervisory}.
The most basic specification is safety, where the system must avoid certain illegal behaviors. 
Safety can be determined based on the current state of the system with some suitable state-space refinement \cite{yin2016synthesis}.
For more complex properties, such as  diagnosability \cite{hu2020design,hu2021diagnosability,cao2024active}, current-state opacity \cite{dubreil2010supervisory,tong2018current,barcelos2021enforcing,moulton2022using} and strong detectability \cite{shu2013enforcing}, their supervisor synthesis problems can also be viewed as safety enforcements on the current-state estimate. In \cite{yin2016uniform}, a unified approach was proposed for synthesizing supervisors by constructing a game structure for a general class of properties that can be evaluated using the current information-state.  
However, this approach cannot be used for enforcing predication-based properties in our work as the information structure in \cite{yin2016uniform} cannot 
handle the future dependency challenge.

\subsubsection{Supervisor Control with Delayed Information}
Some observational properties, such as $K$-step opacity, infinite-step opacity, or delayed detectability, involve determining whether the state status of the system can be inferred after certain information delays. 
While these properties may appear to depend on future information, they can, in fact, be fully resolved using only the current information, as they are essentially smoothing problems. 
Specifically, given an observation sequence, one can construct all delayed state estimates along the trajectory. 
By employing efficient information-state representations, the game-based synthesis approach remains applicable even to these types of properties with delayed information \cite{yin2019supervisory,yin2020synthesis,wintenberg2021enforcement,liu2022enforcement,zhang2023polynomial,xie2024optimal}.

\subsubsection{Notions and Verification of Prediction-Based Properties}
In the literature, several specific notions of prediction-based properties, along with their verifications, have been studied. 
The first observational property involving future information is the notion of predictability or prognosability \cite{watanabe2021fault,chen2014stochastic,ran2022prognosability}. 
This property can be verified in polynomial-time by checking the distinguishability between boundary states and non-indicator states \cite{takai2015robust,you2019verification}.
In contrast to predictability, the notion of pre-opacity captures the requirement that certain critical information must not be predictable too far in advance \cite{yang2022secure,hou2022abstraction}. 
This property can also be efficiently verified using the observer structure.
In such verification problems, the future behavior of the system can be analyzed based on the original plant model, which remains fixed since no control is applied. The observer structure can then be used to examine the (in)distinguishability of states with differing future behaviors. However, such an approach cannot be  extended to synthesis problems, as the future behavior depends on the supervisor being synthesized. 
Furthermore, existing works study prediction-related properties on a case-by-case basis, whereas our work introduces a general framework for prediction-based properties.
 
\subsubsection{Control Synthesis of Prediction-Based Properties} 
To our knowledge, there are only very few works addressing the control synthesis problem for properties involving future information. One exception is \cite{haar2020active}, where the authors study the action prediction problem and propose methods for synthesizing controllers to ensure predictability. They show that this problem can be reduced to B\"{u}chi games. However, compared to our work, the  technique in \cite{haar2020active} is tailored specifically to predictability and requires all controllable events to be observable, a restriction not imposed in our framework. Another related work is \cite{chen2023you}, which investigates unpredictable planning for co-safe linear temporal logic tasks. Nevertheless, \cite{chen2023you} assumes a system model without unobservable events, which significantly simplifies the synthesis challenge. Moreover, their approach is also specialized for unpredictability, unlike the general framework presented in this work.

\subsection{Organization}

The rest of this article is organized as follows: Section~\ref{sec:pre} introduces the necessary preliminaries. In Section~\ref{sec:problem formulation}, we define the prediction-based properties and formulate the control problem for enforcing these properties. Section~\ref{sec: inform-state} presents a novel class of information states for prediction-based properties. In Section~\ref{sec:syhthesis}, we introduce the concept of BTS and propose a synthesis algorithm to design a maximally permissive partial-observation supervisor for enforcing prediction-based properties. Finally, Section~\ref{sec: conc} concludes the article.

\section{Preliminary}\label{sec:pre}
\subsection{System Model}
Let $\Sigma$ be a finite set of events. 
A \emph{string} is a finite sequence of events, and we denote by $\Sigma^*$ the set of all strings over $\Sigma$ including the empty string $\epsilon$. 
For any string $s \in \Sigma^*$, its length  is denoted by $|s|$ with $|\epsilon|=0$. 
For any integer $k$, we denote by $\Sigma^k=\{s\in \Sigma^*\mid |s|=k\}$ 
the set of strings with length $k$. 
A language $L \subseteq \Sigma^*$ is a set of strings. 
For any string $s\in L$, we denote by $L/s$ the post-language of $s$ in $L$, i.e., $L/ s:=\{w\in \Sigma^*\mid sw\in L\}$. 
The prefix-closure of $L$ is denote by $\Bar{L}$, i.e., $\Bar{L}=\{ u \in \Sigma^* \mid \exists v \in \Sigma^* \text{ s.t. } uv \in L \}$.
A language $L\subseteq \Sigma^*$ is said to be \emph{live} if $\forall s\in L,\exists \sigma\in \Sigma: s\sigma\in L$.

We consider a discrete-event system modeled by a deterministic finite-state automaton (DFA)
\[G=(X,\Sigma,\delta,x_0),\] 
where 
$X$ is the finite set of states, 
$\Sigma$ is the finite set of events, 
$\delta:X\times \Sigma \to X$ is the partial deterministic transition function such that 
$\delta(x,\sigma)=x'$ means that there exists a transition from $x$ to $x'$ with event $\sigma$, 
and $x_0 \in X$ is initial state. 
The domain of the transition function $\delta$ can also be extended to 
$\delta: X\times \Sigma^*\rightarrow X$  recursively by: 
for any $x\in X, s\in \Sigma^*, \sigma\in \Sigma$,
we have $\delta(x, s\sigma) = \delta(\delta(x, s), \sigma)$ with $\delta(x, \epsilon)=x$.
The language generated by $G$ from state $x$ is defined by 
$\mathcal{L}(G,x)=\{s\in \Sigma^*\mid\delta(x,s)!\}$, where ``!" means ``is defined". 
We also define $\mathcal{L}(G, Q):=\bigcup_{x\in Q}\mathcal{L}(G,x)$ as the language generated from a set of states $Q\subseteq X$, and define the language generated by $G$ as $\mathcal{L}(G):=\mathcal{L}(G, x_0)$. 
For simplicity, for string $s\in \mathcal{L}(G)$, we write $\delta(x_0,s)$  as $\delta(s)$. 
For technical purposes, we assume that system $G$ is live, i.e., $\forall x\in X,\exists \sigma\in \Sigma: \delta(x,\sigma)!$. 

For partially-observed DES, we assume that the event set is further partitioned as \vspace{-3pt}
\[
\Sigma=\Sigma_o \dot{\cup} \Sigma_{uo},
\]
where $\Sigma_o$ is the set of observable events and $\Sigma_{uo}$ is the set of unobservable events. 
The occurrence of each event is imperfectly observed through a natural projection $P: \Sigma^* \to \Sigma_{o}^{*}$ defined as follows:
\begin{equation}
    P(\epsilon)=\epsilon \text{ and } \begin{aligned}
	P(s\sigma) = 
		\left\{
		\begin{array}{ll}
			 P(s)\sigma & \text{if}\quad    \sigma \in \Sigma_o  \\
			 P(s)                & \text{if}\quad    \sigma \in \Sigma_{uo} 
		\end{array}
		\right..   
\end{aligned} 
\end{equation}
The inverse projection $P^{-1}:\Sigma^{*}_{o} \to 2^{\Sigma^*}$ is defined by $P^{-1}(\alpha):=\{ s \in \mathcal{L}(G)\mid P(s)=\alpha \}$. 
For any observation $\alpha\in P(\mathcal{L}(G))$, 
the \emph{current-state estimate} is the set of all possible states the system could be in currently when $\alpha$ is observed, i.e.,
$\mathcal{E}(\alpha)=\{ \delta(s) \in X \mid s \in P^{-1}(\alpha)\}$.

\subsection{Supervisory Control}
In the supervisory control framework \cite{cassandras2008introduction}, a supervisor can restrict the behavior of the system $G$ by dynamically disabling/enabling some system events. In this setting, the event set $\Sigma$ is further partitioned as 
\[
\Sigma= \Sigma_c \dot{\cup} \Sigma_{uc},
\]
where $\Sigma_c$ is the set of controllable events and $\Sigma_{uc}$ is the set of uncontrollable events. 
A control decision $\gamma\in 2^\Sigma$ is said to be valid if $\Sigma_{uc}\subseteq\gamma$, namely, uncontrollable events can never be disabled.
We define $\Gamma=\{ \gamma \in 2^{\Sigma} \mid \Sigma_{uc} \subseteq \gamma \}$ as the set of valid control decisions. 
Since a supervisor can only make decisions based on its observations, a partial-observation supervisor is a function 
\[
S : P(\mathcal{L}(G)) \to \Gamma.
\]
We use the notation $S/G$ to represent the closed-loop  system under control. 
The language generated by $S /G$, denoted by $\mathcal{L}(S/G)$, is defined recursively as follows:
\begin{itemize}
    \item [1)]$\epsilon\in \mathcal{L}(S/G)$; and
    \item [2)]for any $s\in \Sigma^*$, $\sigma\in\Sigma$, we have 
    $s\sigma\in \mathcal{L}(S/G)$ iff 
    (i) $s\sigma\in \mathcal{L}(G)$, (ii) $s\in \mathcal{L}(S/G)$, and (iii) $\sigma\in S(P(s))$.
\end{itemize}

Note that, when supervisor $S$ is given, upon observing $\alpha\in P(\mathcal{L}(S/G))$, the state estimate is more precise as some strings in the original systems are disabled. 
Formally, we extend the natural projection and state estimate  for the open-loop case to the closed-loop setting by
\begin{align}
     P^{-1}_S(\alpha):=&\{ s \in \mathcal{L}(S/G) \mid P(s)=\alpha \}\nonumber\\
     \mathcal{E}_S(\alpha):=&\{ \delta(s) \in X\mid s \in P^{-1}_S(\alpha)\},\nonumber
\end{align}
where we use a subscript to emphasize that the system behavior is controlled by supervisor $S$. 

Let $q\subseteq X$ be a set of states,  $\gamma \in \Gamma$ be a control decision 
and $\sigma \in \Sigma_o$ be an observable event. 
Then the \emph{unobservable reach} of $q$ under $\gamma$ is defined by 
\[
\text{UR}_\gamma(q)=\{ \delta(x,w) \in X \mid  x \in q, w \in (\Sigma_{uo}\cap \gamma)^* \}.
\]
The \emph{observable reach} of $q$ upon $\sigma\in\Sigma_o$ is defined by  
\[
\text{OR}_{\sigma}(q)=\{ \delta(x,\sigma)\in X \mid  x \in q  \}.
\]

\section{Notions of Prediction-Based Properties}\label{sec:problem formulation}
In this section, we formally provide the general definition of prediction-based properties. 
Then we provide several specific instances of the general definition motivated by the applications of intention security and fault prognosis. Since our focus is on the synthesis problem, all definitions are provided directly for the closed-loop systems.
\subsection{State Predictions and Prediction-Based Properties}
For any   string $s \in \mathcal{L}(S/G)$ in the closed-loop system, we can predict future system states using the plant model $G$ and supervisor $S$. Formally, the $\mathbf{k}$\textbf{-step reachable set} is defined as the set of all possible states reachable in exactly $k$ steps following string $s$, given by:
\begin{equation}\label{eq:k-reachable-set}
    \text{Reach}_k(s)=\{\delta(st)\in X \mid  
    st\!\in\! \mathcal{L}(S/G) \wedge |t|\!=\!k\}.
\end{equation}

In this work, we investigate prediction properties related to whether a system will reach a set of \emph{critical states} $X_C \subseteq X$. 
For each prediction, there are three possible outcomes:  
\begin{itemize}
    \item 
    The system will  reach critical states for sure; 
    \item 
    The system  will not  reach critical states for sure;
    \item 
    It is uncertain whether the system will reach critical states. 
\end{itemize}
To formalize this, we define a \textbf{membership function} over the three-value domain \( \{\texttt{Y}, \texttt{N}, \texttt{U}\} \):  
\[
\chi_C \colon 2^X \to \{\texttt{Y}, \texttt{N}, \texttt{U}\},
\]  
where  for any set of states $q\subseteq X$, we have 
\begin{equation}
\chi_C(q) = 
\begin{cases}
    \texttt{Y} & \text{if } q \subseteq X_C, \\
    \texttt{N} & \text{if } q \cap X_C = \emptyset, \\
    \texttt{U} & \text{otherwise.}
\end{cases}
\end{equation}

By considering \(q\) as \(\text{Reach}_k(s)\), for each current string \(s\), the membership status 
(\texttt{Y}/\texttt{N}/\texttt{U}) may vary at different future instants \(k\). 
To bound our analysis, we assume the property of interest is evaluated over a finite prediction horizon  
\( H \geq 1 \). 
We then encode the membership status over this horizon as a vector.

\begin{mydef}[\bf Prediction Vectors]\upshape\label{def:pre-vec}
Given system $G$ and supervisor $S$, 
for any string $s\in \mathcal{L}(S/G)$, the prediction vector 
(w.r.t.\ critical states $X_C$ and prediction horizon $H$) is a $H+1$ dimensional vector over $\{\texttt{Y}, \texttt{N}, \texttt{U}\}$ defined by
\begin{equation}
   \xi^S(s)=(\xi^S(s)[0],\xi^S(s)[1],\dots, \xi^S(s)[H]) \in \{\texttt{Y}, \texttt{N}, \texttt{U}\}^{H+1},
\end{equation}
where for each $k=0,1,\dots,H$, 
we  have 
\begin{equation}
    \xi^S(s)[k]= \chi_C(\text{Reach}_k(s)).
\end{equation}
We denote by $\V=\{\texttt{Y}, \texttt{N}, \texttt{U}\}^{H+1}$ the set of all prediction vectors with horizon $H$.\hfill$\blacksquare$
\end{mydef}

The prediction vector described above is generated based on the system's actual executed string. However, in partial observation settings, where the actual string is not directly observable, predictions must instead rely on the inverse projection of observed events. To formally characterize prediction under partial observation, 
we assume that the system updates its prediction \emph{immediately} upon each new observable event occurrence. 
To this end, we define 
\begin{equation}
    \mathcal{L}_o(S/G):= (\mathcal{L}(S/G) \cap \Sigma^*\Sigma_o) \cup \{ \epsilon \}
\end{equation}
as the set of strings that end up with observable events including the empty string. Then  for each observation $\alpha \in P(\mathcal{L}(S/G))$, we define 
\[
\O(\alpha)=\{ s \in \mathcal{L}_o(S/G): P(s)=\alpha \}  
\]
as the set of observationally equivalent strings that terminate with observable events, 
which is the set of possible strings immediately when the system observes $\alpha$. These strings serve as the starting points for predictions. Since such observationally equivalent strings are generally not unique, the prediction vectors belong to a set. 

\begin{mydef}[\bf Prediction Sets]\upshape\label{def:pre-set}
Given system $G$ and supervisor $S$, 
for any observation $\alpha \in P(\L(S/G))$, the prediction set
(w.r.t.\ critical states $X_C$ and prediction horizon $H$), denoted by $\Xi^S(\alpha)$,   
is the set of prediction vectors for strings in $ s\in \mathcal{O}(\alpha)$, i.e., 
\begin{equation}
    \Xi^S(\alpha)=\{  \xi^S(s) \in \V:   s\in \mathcal{O}(\alpha)\}.
\end{equation}
We denote by $\mathbf{\Xi}\!=\!2^{\V}$ the set of  all possible prediction sets. \hfill$\blacksquare$
\end{mydef} 

\begin{figure}
  \centering
    \subfigure[$G$\label{fig:G}]{\centering   \begin{tikzpicture}[->,>={Latex}, thick, initial text={}, node distance=1cm, initial where=above, thick, base node/.style={circle, draw, minimum size=5mm,font=\footnotesize}]  
   \node[state, initial, base node, ] (0) {$0$};
   \node[state, base node, ] (1) [left of=0] {$1$};
   \node[state, base node, ] (2) [below of=1] {$2$};
   \node[state, base node, ] (3) [below of=0] {$3$};
   \node[state, base node, ] (4) [below of=3] {$4$};
   \node[state, base node, ] (5) [below of=2] {$5$};
   \node[state, base node, ] (6) [below of=4] {$6$};
   \node[state, base node, fill=red] (7) [below of=5] {$7$};
   
   \path[->]
   (0) edge node [yshift=0.2cm] {$a$} (1)
   (1) edge node [xshift=-0.2cm] {$b$} (2)
   (2) edge node [xshift=-0.2cm] {$o_2$} (5)
   (5) edge node [xshift=-0.2cm] {$e$} (7)
   (0) edge node [xshift=-0.2cm] {$c$} (3)
   (3) edge node [xshift=-0.2cm] {$b$} (4)
   (3) edge node [yshift=0.2cm] {$o_1$} (5)
   (4) edge node [xshift=-0.2cm] {$o_1$} (6)
   (1) edge [loop left ] node         [xshift=0.2cm, yshift=0.3cm]     {$d$} ()
   (2) edge [loop left ] node         [xshift=0.2cm, yshift=0.3cm]     {$o_1$} ()
   (6) edge [loop below ] node         [xshift=-0.4cm, yshift=0.3cm]     {$o_1$} ();
   \draw[->]
   (7) to [bend left=45] node [xshift=-0.2cm, yshift=0.2cm] {$o_1$} (5);
   \end{tikzpicture}}
    \subfigure[$S_1/G$\label{fig:G1}]{\centering   \begin{tikzpicture}[->,>={Latex}, thick, initial text={}, node distance=1cm, initial where=above, thick, base node/.style={circle, draw, minimum size=5mm,font=\footnotesize}]  
   \node[state, initial, base node, ] (0) {$0$};
   \node[state, base node, ] (3) [below of=0] {$3$};
   \node[state, base node, ] (4) [below of=3] {$4$};
   \node[state, base node, ] (5) [left of=4] {$5$};
   \node[state, base node, ] (6) [below of=4] {$6$};
   \node[state, base node, fill=red] (7) [below of=5] {$7$};

   \draw[->, dashed] (0) -- node [yshift=0.2cm] {$a$}(-1cm, 0cm);
   
   \path[->]
   (5) edge node [xshift=-0.2cm] {$e$} (7)
   (0) edge node [xshift=-0.2cm] {$c$} (3)
   (3) edge node [xshift=-0.2cm] {$b$} (4)
   (3) edge node [yshift=0.2cm] {$o_1$} (5)
   (4) edge node [xshift=-0.2cm] {$o_1$} (6)
   (6) edge [loop below ] node         [xshift=-0.4cm, yshift=0.3cm]     {$o_1$} ();
   \draw[->]
   (7) to [bend left=45] node [xshift=-0.2cm, yshift=0.2cm] {$o_1$} (5);
   \end{tikzpicture}}  
    \subfigure[$S_2/G$\label{fig:G2}]{\centering   \begin{tikzpicture}[->,>={Latex}, thick, initial text={}, node distance=1cm, initial where=above, thick, base node/.style={circle, draw, minimum size=5mm,font=\footnotesize}]  
   \node[state, initial, base node, ] (0) {$0$};
   \node[state, base node, ] (1) [left of=0] {$1$};

   \draw[->, dashed] (0) -- node [xshift=-0.2cm] {$c$}(0cm, -1cm);
   \draw[->, dashed] (1) -- node [xshift=-0.2cm] {$b$}(-1cm, -1cm);
   
   \path[->]
   (0) edge node [yshift=0.2cm] {$a$} (1)
   (1) edge [loop left ] node         [xshift=0.2cm, yshift=0.3cm]     {$d$} ()
   ;
   \end{tikzpicture}}\\
  \caption{For   $G$, we have $\Sigma_{o}=\{o_1,o_2\}$, $\Sigma_{c}=\{a,b,c\}$ and $X_S=\{7\}$.}
	\label{fig:example1}
\end{figure}
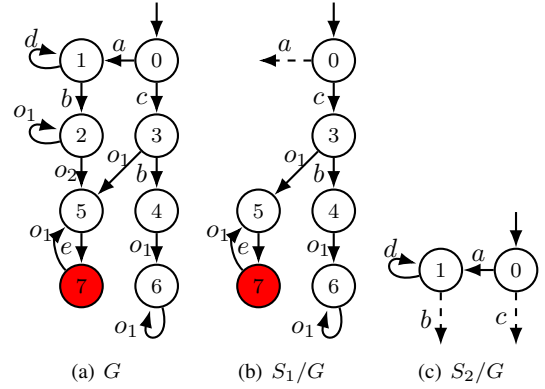
We illustrate the above concepts with the following example.
\begin{myexm}\label{example1}
Let us consider system $G$ shown in Figure~\ref{fig:G} with observable events $\Sigma_o = \{o_1, o_2\}$. 
First, we assume that supervisor $S$ disables nothing, i.e., $\mathcal{L}(G)=\mathcal{L}(S/G)$. 
For string $s=abo_1$, we have 
\[
\text{Reach}_{0}(s)\!=\! \{2\},  
\text{Reach}_{1}(s) \!=\! \{2,5\}, 
\text{Reach}_{2}(s) \!= \!\{2,5,7\}.
\]
Suppose that the critical states of interest are $X_C= \{7\}$. 
Then the prediction vector for string $s=abo_1$ with  $H=2$ is 
\[
\xi^S(abo_1)=( \texttt{N},\texttt{N}, \texttt{U} ).
\]
Note that the system can only observe $\alpha=P(abo_1)=o_1$, 
and it is also possible that the actual string is $t=co_1\in \O(o_1)$. 
For this string, we have
\[
\text{Reach}_{0}(t)\!=\! \{5\},  
\text{Reach}_{1}(t) \!=\! \{7\},   
\text{Reach}_{2}(t) \!= \!\{5\}, 
\]
and the prediction vector for string $t=abo_1$ with $H=2$ is 
\[
\xi^S(co_1)=( \texttt{N},\texttt{Y}, \texttt{N} ).
\]
Similarly, for string $cbo_1\in \O(o_1)$, 
we have $\xi^S(cbo_1)=( \texttt{N},\texttt{N}, \texttt{N} )$.
Overall, for observation $\alpha=o_1$, its predication set w.r.t.\ $X_C= \{7\}$
is 
\begin{equation}\label{eq:xio1}
    \Xi^S(o_1)=\{( \texttt{N},\texttt{Y}, \texttt{N} ),( \texttt{N},\texttt{N}, \texttt{U} ),( \texttt{N},\texttt{N}, \texttt{N}) \}.
\end{equation}
 
\end{myexm}

Therefore, upon each string $s\in \mathcal{L}_o(S/G)$, the system observes $P(s)$, and estimates all possible strings $\O(P(s))$ as well as predicting their membership vectors $\Xi^S(P(s))$. 
To capture the desired future information patterns, 
we use a generic predicate on the predication set $\Xi^S(P(s))$ to evaluate its satisfication status.

\begin{mydef}[\bf Evaluation Functions]\upshape
An evaluation function is a predicate on predication sets of the form
\begin{equation}
    \Phi: \mathbf{\Xi}\to \{0,1\}.\vspace{-6pt}
\end{equation}\hfill$\blacksquare$
\end{mydef}

Now, we are ready to formally introduce the definition of prediction-based properties. 

\begin{mydef}[\bf Prediction-Based Properties]\upshape
A prediction-based  property is a tuple $(X_C,\Phi)$,  
where 
$X_C\subseteq X$ is a set of critical states 
and 
$\Phi: \mathbf{\Xi}\to \{0,1\}$ is an evaluation function. 
Furthermore, given system $G$ and supervisor $S$, we say 
\begin{itemize}
    \item 
    an observation  $\alpha\in P(\mathcal{L}(S/G))$ satisfies $(X_C,\Phi)$, 
    denoted by $\alpha\models (X_C,\Phi)$, if its prediction set satisfies the evaluation function, i.e.,  
    \begin{equation}
        \Phi(  \Xi^S( \alpha )   )=1.
    \end{equation}
    \item 
    the closed-loop system $S/G$ satisfies $(X_C,\Phi)$, 
    denoted by $S/G \models (X_C,\Phi)$, if all observations satisfy $(X_C,\Phi)$, i.e., 
     \begin{equation}
      \forall \alpha\in P(\mathcal{L}(S/G)):   \Phi(  \Xi^S( \alpha )   )=1.\vspace{-6pt}
     \end{equation}\hfill$\blacksquare$
\end{itemize}
\end{mydef}
\begin{myexm}
We continue to consider the running example shown in Figure~\ref{fig:G} with \( H = 2 \).
Let \(\Phi\) be an evaluation function defined as follows: 
for any \(\Xi^S \in 2^\V\), we have \(\Phi(\Xi^S) = 0\) if and only if 
\(\exists k \in \{0,1,2\}, \forall \xi^S \in \Xi^S: \xi^S[k] = \texttt{Y}\). 
In other words, an observation \(\alpha\) satisfies \((X_C, \Phi)\) if, at any future instant within the next two steps, it cannot definitively determine that the system will reach \(X_C\) at that specific time. As we will elaborate in Section~\ref{sec:3.2}, this property captures the notion of pre-opacity.  

Clearly, if the supervisor disables nothing, 
then for observation \(\alpha = o_1\), we have \(o_1 \models (X_C, \Phi)\). This holds because, for the prediction set \(\Xi^S(o_1)\) in \eqref{eq:xio1}, at each time instant \(k = 0,1,2\), there exists at least one prediction vector whose corresponding component is not \(\texttt{Y}\).  
However, for observation $\beta=o_2$, we have $\mathcal{O}(o_2) = \{ad^n bo_2 \mid n \in \mathbb{N}\}$. 
For each possible string $s \in \mathcal{O}(o_2)$, 
its reachable sets are
    \[
    \text{Reach}_{0}(s)=\text{Reach}_{2}(s)  = \{5\}, \text{Reach}_{1}(s) = \{7\},
    \]
and the prediction set for $o_2$ is singleton $\Xi^S(o_2) = \{(\texttt{N}, \texttt{Y}, \texttt{N})\}$. 
Therefore,   \(o_2 \not\models (X_C, \Phi)\), which means that 
$S/G\not\models (X_C, \Phi)$ when the supervisor $S$ disables nothing.  
\end{myexm}

Our objective is to synthesize a partial observation supervisor such that the closed-loop system satisfies some given prediction-based properties. 
\begin{myprob}[\bf Supervisor Synthesis for Prediction-Based Properties]\label{problem}\upshape
Given system $G$ with controllable events $\Sigma_c$, observable events $\Sigma_o$ and 
prediction-based property $(X_C,\Phi)$, 
synthesize a partial-observation supervisor 
$S:\Sigma_o^*\to \Gamma$
such that $\mathcal{L}(S/G)$ is live and 
$S/G\models (X_C,\Phi)$. 
\end{myprob}

\begin{myexm}
We continue to consider the running example and assume that $\Sigma_c=\{a,b,c\}$ is the set of controllable events.  
To enforce prediction-based property $(X_C,\Phi)$ as specified in the previous example, 
there are two incomparable solutions. 
One possibility is to use a supervisor $S_1$ that disables event $a$ at the very beginning, 
whose closed-loop behavior $S_1/G$ is shown in Figure~\ref{fig:G1}. 
To see this, one can easily check that in $S_1/G$, we have  $\Xi^{S_1}(\epsilon)=\{(\texttt{N},\texttt{N},\texttt{N})\}$ and $\Xi^{S_1}(o_1^n)=\{(\texttt{N},\texttt{Y},\texttt{N}),(\texttt{N},\texttt{N},\texttt{N})\}$ for all $n\in\mathbb{N}^+$. 
Another solution is to use a supervisor $S_2$ that disables events $b$ and $c$ at the very beginning, 
whose closed-loop behavior $S_2/G$ is shown in Figure~\ref{fig:G1}. 
For this system, since it will not even reach critical state $7$, it clearly satisfies $(X_C,\Phi)$.
\end{myexm}

\subsection{Applications of Prediction-Based Properties}\label{sec:3.2}
The definition of prediction-based properties above is intentionally generic. 
In practice, the critical set \( X_C \) and evaluation function \( \Phi \) can be instantiated to define specific prediction-based properties tailored to different applications.
Below, we introduce two concrete examples of such properties, which will also help illustrate the previously introduced notations.

\subsubsection{Pre-Opacity for Intention Security} 
First, we consider the notion of pre-opacity proposed by \cite{yang2022secure}, which ensures that a system always maintains plausible deniability regarding the intention to reach certain secret states in the future.  
Specifically, we assume the presence of an intruder modeled as a passive observer with access to the observable event set \(\Sigma_o\). At any given time, the intruder can predict the system’s reachable states within a finite horizon of \(M\) steps. The security requirement is that whenever the system enters a secret state, the intruder must never be able to definitively predict this visit \(K\)-steps in advance within its prediction horizon. This requirement can be captured in terms of prediction-based property as follows.

\begin{mydef}[\bf ($M,K$)-Pre-Opacity]\upshape\label{def:pre-opa}
Given system $G$ with observable events $\Sigma_o$, supervisor $S$, 
a set of secret states $X_S\subset X$, 
and two non-negative integers $K, M \in \mathbb{N}, K\leq M$, 
we say the closed-loop system $S/G$ is ($M,K$)-pre-opaque if  
$S/G\models (X_S, \Phi_{opa})$, where 
for each $\Xi^S\in 2^\V$, we have
\begin{equation}
\Phi_{opa}(\Xi^S)=1
\quad
\Leftrightarrow
\quad
\bigwedge_{k=K,\dots,M} (\bigvee_{\xi^S\in \Xi^S}  \xi^S[k]\neq\texttt{Y} )\vspace{-12pt}
\end{equation}\hfill$\blacksquare$
\end{mydef}

Intuitively, if \( S/G \not\models (X_S, \Phi_{opa}) \), then there exists an observation \(\alpha \in P(\mathcal{L}(S/G))\) and a future instant \( k \in \{K, \dots, M\} \) such that for every possible string \( s \in \mathcal{O}(\alpha) \), the system is guaranteed to visit a secret state in \( X_S \) at time \( k \). Consequently, the system's intention to reach a secret state at a specific future time can be precisely inferred more than \( K \) steps in advance.
Therefore, the evaluation function $\Phi$ in our running example is essentially $(2,0)$-pre-opacity by considering $X_S$ as the critical states set $X_C=\{7\}$.

\subsubsection{Predictability for Fault Prognosis} 
The above notion of pre-opacity ensures that an intruder \emph{cannot} predict the system's intention to reach secret states. However, in certain applications such as fault prognosis, the system must instead guarantee that its execution of critical behaviors (e.g., faults) \emph{can be} predicted in advance. This requirement, known as predictability \cite{genc2009predictability}, has been extensively studied in the literature.

To formalize this, we assume that the state space of the system is partitioned as  
\[  
X = X_N \dot{\cup} X_F,  
\]  
where \( X_N \) represents the set of normal states and \( X_F \) denotes the set of fault states, a fault is a transition from a state in \(X_N\) to a state in \(X_F\). Additionally, we assume that faults are permanent in the sense that once the system enters a fault state, it remains in a fault state indefinitely, i.e.,  
\[  
\forall x \in X_F, \forall s \in \mathcal{L}(G, x): \delta(x, s) \in X_F.  
\]  
To quantify the performance of a predictor, two performance bounds are considered in the literature \cite{yin2016decentralized}: 
\begin{itemize}
    \item 
    No missed alarm: any visit to fault states can be predicted $K$ steps ahead; 
    \item 
    No false alarm: once a fault alarm is issued, the system will visit fault states for sure within $M$ steps. 
\end{itemize}
Such a requirement can also be captured in terms of prediction-based property as follows. 

\begin{mydef}[\bf ($M,K$)-Predictability]\upshape\label{def:pre-predi}
Given system $G$ with observable events $\Sigma_o$, supervisor $S$, 
a set of fault states $X_F\subset X$, 
and two non-negative integers $K, M \in \mathbb{N}, K\leq M$, 
we say the closed-loop system $S/G$ is ($M,K$)-predictable if  
$S/G\models (X_F, \Phi_{pre})$, where 
for each $\Xi^S\in 2^\V$, we have
\begin{equation}
\Phi_{pre}(\Xi^S)=1
\ 
\Leftrightarrow \
( \!\!\!\!\bigwedge_{\xi^S\in \Xi^S} \!\! \!\!\xi^S[K]=\texttt{N} )
\vee
(\!\!\!\!\bigwedge_{\xi^S \in \Xi^S} \!\! \!\!\xi^S[M]=\texttt{Y}).\vspace{-12pt}
\end{equation}\hfill$\blacksquare$
\end{mydef}

To clarify the definition further, suppose the system fails to be ($M,K$)-predictable. 
Then there must exist two distinct prediction vectors $\xi^S, \xi'^S \in \Xi^S$, 
where $\xi^S[k] \neq \texttt{N}$ holds for some $k \leq K$, indicating possible fault occurrence within $K$ steps, 
while $\xi'^S[M] \neq \texttt{Y}$ shows the fault is not guaranteed within $M$ steps. 
This creates an inherent conflict: though faults may appear imminent in the near-term ($K$-step) prediction, 
their inevitability cannot be confirmed within the longer $M$-step window. 
As a result, no predictor can simultaneously avoid both missed alarms and false alarms when processing such observations, 
since the short-term possibility for fault behaviors ($\xi^S$) contradicts the long-term uncertainty for normal behaviors ($\xi'^S$).

Conversely, if the system is indeed ($M,K$)-predictable, then for any string $s \in \mathcal{L}(S/G)$ where $\delta(s t) \in X_F$ holds for some $|t| = K$ (indicating a potential fault within $K$ steps), 
let $\Xi^S$ be the prediction set for the observation $P(s)$. 
In this case, we know $\xi^S(s)[K] \neq \texttt{N}$ must hold for all $\xi^S(s) \in \Xi^S$. 
This implies $\bigwedge_{\xi^S \in \Xi^S(P(s))} \xi^S[M] = \texttt{Y}$, meaning the fault is guaranteed to occur within $M$ steps.
Consequently, the system allows reliable prediction of faults at least $K$ steps in advance, with the assurance that any predicted fault will inevitably occur within the subsequent $M$-step horizon.

\subsection{Prediction-Based Properties with Multiple Regions}
While the above definition of prediction-based properties is formulated for a single critical state set $X_C$, many practical applications require evaluating prediction correctness across multiple critical regions $X_1, X_2, \dots, X_m \subseteq X$. Our framework can be naturally extended to this more general setting through straightforward modifications to accommodate multiple critical sets.

Formally, in this setting,  for each string $s\in \mathcal{L}(S/G)$, 
the multiple prediction vector is a tuple
\begin{equation}
    \vec{\xi}^S(s)=(\xi^S_1(s),\xi^S_2(s),\dots, \xi^S_m(s))\in \V^m, 
\end{equation}
where each $\xi^S_i(s)$ is the previous defined prediction vector w.r.t. critical region $X_i$. 
For each $\alpha\in P(\mathcal{L}(S/G))$, the multiple prediction set is 
\begin{equation}
    \vec{\Xi}^S(\alpha)=\{ \vec{\xi}^S(s)\mid s\in\O(\alpha) \}\in 2^{\V^m}.
\end{equation}
Then the multiple prediction-based property is a tuple 
$(\{X_i\}_{i=1}^m,\Phi)$, 
where the evaluation function is extended to
\begin{equation}
    \Phi: 2^{\V^m}\to \{0,1\}.
\end{equation}
We use the following notion, also motivated by security considerations, 
to illustrate the extension to multiple regions. 

\begin{mydef}[\bf ($M,K,m$)-Anonymity]\upshape\label{def:pre-ano}
Given system $G$ with observable events $\Sigma_o$, supervisor $S$, 
$m$ disjoint critical regions $X_{1},\dots, X_{m}\subset X$, 
non-negative integers $K, M \in \mathbb{N}, K\leq M$,  
we say the closed-loop system $S/G$ is ($M,K,m$)-anonymous if  
$S/G\models (\{X_{i}\}_{i=,1\dots,m}, \Phi_{ano})$, where 
$ \Phi_{ano}$ is defined by: 
for each $\vec{\Xi}^S=\{\vec{\xi}^S_{(1)},\dots, \vec{\xi}^S_{(q)}\} \in 2^{\V^m}$, we have
$\Phi_{ano}(\vec{\Xi}^S)=1$, if and only if, 
for each instant $k=K, \dots,M$,  
whenever $\vec{\xi}^{S}_{(j),i}[k]\neq\texttt{N}$ for some $j\leq q, i\leq m$,  
there must exist a set of indices $(j_1,1),\dots, (j_{i-1},i-1),(j_{i+1},i+1),\dots, (j_m,m)$ 
such that 
(i) \(\vec{\xi}^S_{(j_u)}\in \vec{\Xi}^S,\forall u=1,\dots, i-1,i+1,\dots, m\); and
(ii)
    $\vec{\xi}^S_{(j_u),u}[k]\neq\texttt{N},\forall u=1,\dots, i-1,i+1,\dots, m$. \hfill$\blacksquare$
\end{mydef}

This definition also relates to security scenarios involving opacity, where an intruder can predict system behavior for at most \( M \) steps. The key requirement is that the system must maintain ambiguity about visiting any critical region for at least \( K \) steps up to \( M \) steps. Specifically, whenever the system might visit a critical region at a future instant beyond \( K \) steps, it must preserve plausible deniability by showing potential visits to \( m-1 \) other critical regions at the same time instant.  
This concept aligns with existing notions of  anonymity in the literature \cite{sweeney2002k}. 
However, while prior work focuses on current-state ambiguity, our framework extends this principle to future behavior uncertainty.

Note that this definition cannot be captured by a single prediction vector, as we must track membership across all critical regions. 
However, such an extension mainly expands the coding space for reachable sets while preserving the underlying framework. 
Consequently, all techniques developed for the single-region case can be readily adapted to multiple regions. 
For simplicity, hereafter in this work. We will focus solely on the single-region case.

\section{Information Structure with Previewed Predictions}\label{sec: inform-state}
In this section, we first discuss the fundamental challenges in applying conventional partial-observation supervisor synthesis techniques to prediction-based properties. We then present a new information structure that addresses these limitations through a novel information preview mechanism.
\subsection{Challenges in Prediction-Based Properties Synthesis}

Recall that a partial-observation supervisor $S:\Sigma_o^*\to \Gamma$ works as follows:
\begin{itemize}
    \item 
    Initially, the supervisor makes a control decision $\gamma_0\in \Gamma$ from the initial state $x_0$, 
    and   the system evolves unobservably through events in $\gamma_0\cap \Sigma_{uo}$ and reaches possible states 
    $\hat{q}_0=\UR_{\gamma_0}(q_0)$, where $q_0=\{x_0\}$; 
    \item 
    The supervisor then observes an event $\sigma_1\in \gamma_0\cap \Sigma_o$
    and updates its state estimate (without unobservable tails) to ${q}_1=\OR_{\sigma_1}(\hat{q}_0)$; 
    \item 
    Then the supervisor updates its control decision to $\gamma_1\in \Gamma$, 
    updates the unobservable reach   to $\hat{q}_1=\UR_{\gamma_1}(q_1)$, 
    and waits for the next observable event  $\sigma_2\in \gamma_1\cap \Sigma_o$; 
    \item 
    The above recursive procedure is repeated indefinitely, 
    which induces an information-flow
    \begin{equation}\label{eq:is-current}
     q_0\xrightarrow{\gamma_0}\hat{q}_0 \xrightarrow{\sigma_1}q_1\xrightarrow{\gamma_1}\cdots\xrightarrow{\sigma_n}{q}_n \xrightarrow{\gamma_n}\hat{q}_n,  
    \end{equation}
    where  $\hat{q}_i=\UR_{\gamma_i}(q_i)$ and ${q}_{i}=\OR_{\sigma_i}(\hat{q}_{i-1})$. 
\end{itemize}

The recursive process described above is often referred to as the weak version of \emph{separation principle} between control and observations \cite{barrett2000separation,kumar2015stochastic}. This principle states that the current-state estimate of the closed-loop system depends solely on the actual execution history \(\gamma_0\sigma_1\gamma_1...\sigma_n\gamma_n\), remaining independent of the future control policy \( S \). 
Consequently, the power set \( 2^X \) can serve as the set of \emph{information states} 
and  various system properties, such as safety, opacity, and distinguishability, can be evaluated based solely on these information states \( q_i \) or \( \hat{q}_i \) \cite{yin2016uniform}. 
Therefore, for supervisor synthesis for such (current) information-state-based properties, it suffices to search through the information-state space while avoiding states that violate the desired property.

However, this standard approach for partial-observation supervisor synthesis fails in our setting  as the separation principle no longer holds. 
The fundamental issue comes from the inherent dependency between current predictions and the control policy. 
Specifically, to check whether a prediction-based property holds upon an observation, we must compute the reachable set from each possible current state. 
This computation is straightforward for verification problems, where the system dynamics \( G \) are fixed. 
However, it is  problematic for control synthesis. 
The main challenge arises because \textbf{future behavior of the system depends on future control decisions which have not yet been synthesized}.   
This \emph{future-dependency} issue creates a fundamental challenge distinct from existing control synthesis problems, as the predictions being verified depend on control actions that are themselves part of the synthesis objective.
 
\subsection{Preview of Prediction Vectors}
To address the above discussed challenge,  our approach is to augment the information-state space by estimating 
the set of states augmented with some previewed future information rather than the original states. 
Formally, an \textbf{augmented state} is a tuple 
\begin{equation}
(x,\mathbf{v})=(x,\mathbf{v}[0],\mathbf{v}[1], \dots, \mathbf{v}[H]) \in X\times \V, 
\end{equation}
which is a system state augmented with a prediction vector. 
Then we choose the estimates of augmented states as information states to solve our problem. 
\begin{mydef}[\bf Information States]\upshape 
An information state $\imath\in 2^{X\times \V}$ is a set of augmented states  such that 
\[
\forall (x,\mathbf{v}),(x',\mathbf{v}')\in \imath:x=x'\Rightarrow \mathbf{v}=\mathbf{v}'.
\]  
We denote by $\mathbb{I}\subseteq 2^{X\times \V}$ the set of all possible information states satisfying the above condition.  \hfill$\blacksquare$
\end{mydef}

To explain the above definition in more detail, for each information state $\imath\in \mathbb{I}$, 
we define 
\begin{align}
    \state(\imath) =&\{x\in X\mid \exists \mathbf{v}\in \V\text{ s.t. }(x,\mathbf{v})\in \imath\}\\
    \vecc(\imath) =&\{\mathbf{v}\in \V \mid \exists x\in X \text{ s.t. }(x,\mathbf{v})\in \imath\}
\end{align}
as its system state component and its prediction vector component, respectively. 
Intuitively, in each information state, a system state can be augmented with at most one prediction vector. 
This requirement comes from the fact that if the supervisor makes control decisions based on the information state, as will be elaborated later, then any two identical plant states within the same information state must exhibit the same future behavior, regardless of the strings leading to the states. 
Moreover, we will later prove that this restriction does not lose generality for the purpose of control synthesis.
Hence, for each $x\in \state(\imath)$, we denote by $\mathbf{v}_{\imath}^x$ the unique prediction vector augmented with $x$ in $\imath$, 
i.e., $(x,\mathbf{v}_{\imath}^x)\in \imath$. 

Our objective here is to use an information state  $\imath$ to summarize all relevant state information immediately after observing a new event, such as the role of $q_i$ in Eq.~\eqref{eq:is-current}. However, since future control decisions remain undetermined, we cannot precisely ascertain the membership status of each state at future instants based solely on past observations and control decisions. Thus, the augmented prediction vector associated with each state in the information state serves as a \textbf{preview  of future information} by effectively ``borrowing" membership status from future instants. 
For this preview to be meaningful, it must remain consistent with the actual membership facts that will later materialize. While this requires multi-step information consistency, we can enforce it through a one-step consistency condition that applies globally.  
Now, we formalize this idea as  follows. 

\begin{mydef}[\bf Single-Observation Information Patterns]\upshape\label{def:sig-obs-inf-patt}
Let $\imath\in \mathbb{I}$ be an information state and $\gamma\in \Gamma$ be a control decision applied currently.  A \emph{single-observation information pattern} of  $\imath$  under $\gamma$ is a tuple of form
\[
\mathcal{I}=(\imath_u,\{\imath_\sigma\}_{\sigma\in \Sigma_o\cap \gamma}) \in 
\mathbb{I}\times\underbrace{\mathbb{I}\times\cdots \times\mathbb{I}}_{|\Sigma_o\cap \gamma|\text{ times}} 
\]  
such that 
(i)  $\imath\subseteq \imath_u$; and 
(ii) $\state( \imath_u) = \UR_\gamma(  \state(\imath)  )$; and 
(ii)  $ \forall \sigma \in  \Sigma_o\cap \gamma:\state(\imath_\sigma)= \OR_\sigma(  \state( \imath_u  )  )$.\hfill$\blacksquare$
\end{mydef}

Intuitively, a single-observation information pattern $\mathcal{I}$ consists of two types of information states:  
\begin{itemize}
    \item 
    $\imath_u$ represents the information state reached unobservably from $\imath$ under control decision $\gamma$, 
    and therefore, $\imath$ needs to be included in $\imath_u$; and 
    \item 
    $\imath_\sigma$ represents the information state reached immediately from $\imath_u$ after observing event $\sigma$.  
\end{itemize}
Note that for a given information state $\imath$ and control decision $\gamma$, the single-observation information pattern is  not unique in general. This is because different prediction vectors can be assigned to the system states in $\state(\imath_u)$ and $\state(\imath_\sigma)$. However, not all such assignments are meaningful. As previously discussed, 
the previewed information (encoded as prediction vectors) must remain consistent with the actual future behavior.  
This is formalized as follows. 

\begin{mydef}[\bf One-Step Reachable Sets]\upshape \label{def:one-step-rea-set}
Let $\mathcal{I}=(\imath_u,\{\imath_\sigma\}_{\sigma\in \Sigma_o\cap \gamma})$ be 
a single-observation information pattern of  $\imath$  under $\gamma$. 
For augmented state $\tilde{x}=(x, \mathbf{v})\in \imath_u$, 
its \emph{one-step reachable set} within $\mathcal{I}$ is defined by 
\begin{align}
\mathcal{R}_{\mathcal{I}}(\tilde{x})=&
\{(x',\mathbf{v'})\mid   \sigma\!\in\! \Sigma_{uo}\!\cap\! \gamma, x'\!=\!\delta(x,\sigma), \mathbf{v'}\!=\!\mathbf{v}^{x'}_{\imath_u}\} \nonumber\\
\cup&
\{(x',\mathbf{v'})\mid   \sigma\!\in\! \Sigma_{o}\!\cap\! \gamma, x'\!=\!\delta(x,\sigma), \mathbf{v'}\!=\!\mathbf{v}^{x'}_{\imath_\sigma}\}, 
\end{align}
which is the set of augmented states that can be reached from $\tilde{x}$ in one step 
either in $\imath_u$ through an unobservable event
or in $\imath_\sigma$ through an observable event $\sigma$. \hfill$\blacksquare$
\end{mydef}

To enforce information consistency on prediction vectors,  
for each augmented state in $\imath_u$, its committed membership status at time instant $k$ must align with the membership status of augmented states in its one-step reachable sets at time instant $k-1$.  
For instance, if it is asserted that critical states will certainly be reached in $k$ steps from the current state, then for every subsequent state reached in the next step,   it must also be asserted that critical states will certainly be reached in $k-1$ steps thereafter.  
This idea is formalized by the notion of information consistency, defined as follows.

\begin{mydef}[\bf Information Consistency]\upshape \label{def:inf-consis}
Let $\mathcal{I}=(\imath_u,\{\imath_\sigma\}_{\sigma\in \Sigma_o\cap \gamma})$ be 
a single-observation information pattern of $\imath$  under $\gamma$.   
We say augmented state $\tilde{x}=(x, \mathbf{v})\in \imath_u$ is \emph{consistent} (w.r.t.\ critical states $X_C$) if it satisfies the following conditions:
\begin{itemize}
    \item 
    For the current instant, we have
\begin{equation}\label{eq:curr-consis}
  \mathbf{v}[0]=  		
  \left\{
		\begin{array}{ll}
			 \texttt{Y}  & \text{if}\quad    x\in X_C  \\
			 \texttt{N}        & \text{if}\quad   x\not\in X_C
		\end{array}
		\right.  
\end{equation}
    \item 
    For each future instant $k=1,2,\dots, H$, we have
    \begin{equation}\label{eq:one-step-consis}
  \mathbf{v}[k]=  		
  \left\{
		\begin{array}{ll}
			 \texttt{Y}  & \text{if}\quad \forall (x',\mathbf{v'})\in \mathcal{R}_{\mathcal{I}}(\tilde{x}): \mathbf{v'}[k-1]=\texttt{Y}  \\
			 \texttt{N}  & \text{if}\quad \forall (x',\mathbf{v'})\in \mathcal{R}_{\mathcal{I}}(\tilde{x}): \mathbf{v'}[k-1]=\texttt{N}\\
			 \texttt{U}  & \text{otherwise} 
		\end{array}
		\right.  
\end{equation}
\end{itemize}
We say a single-observation information pattern $\mathcal{I}$ is consistent if each augmented state $(x, \mathbf{v})\in \imath_u$ in it is consistent. 
For  information state $\imath\in \mathbb{I}$ and control decision $\gamma\in \Gamma$, 
we  denote by $\mathbb{A}(\imath, \gamma)$ the set of all single-observation information patterns for $\imath$ under $\gamma$ that are consistent, 
and define $\mathbb{A}=\mathbb{A}_{\imath\in \mathbb{I},\gamma\in \Gamma}(\imath, \gamma)$. \hfill$\blacksquare$
\end{mydef} 

We illustrate the above concepts with some examples.  
We first demonstrate through the following example that arbitrary assignment of prediction vectors may lead to ill-defined information patterns, as the asserted future behavior could become unrealizable.

\begin{myexm}[\bf Inconsistent Information Patterns]
We continue to consider the running example shown in Figure~\ref{fig:G} with \( H = 2 \).
Let us first consider a possible information state
$\imath=\{   (0,(\texttt{N},\texttt{N},\texttt{Y}))  \}$, 
which means that one knows for sure that the system is at the initial state, 
i.e., $\state(\imath)=\{0\}$, 
and for state $0$, one asserts that the system is not currently in $X_C$ and will reach $X_C$ for sure in two steps. 
Let us consider control decision $\gamma=\Sigma\setminus\{a\}$, i.e., the supervisor only disables event $a$. 
Then there is no  consistent single-observation information pattern for $\imath$ under $\gamma$. 
To see this, suppose that 
$(\imath_u,\{\imath_\sigma\}_{\sigma\in \Sigma_o\cap \gamma})$ is a consistent  single-observation information patterns. 
According to Def.~\ref{def:sig-obs-inf-patt}, 
we know that 
$\state(\imath_u)=\UR_\gamma(\{0\})=\{0,3,4\}$ 
and 
$\imath_u$ is in the form of 
$\imath_u=\{ (0,(\texttt{N},\texttt{N},\texttt{Y})), (3,\mathbf{v}), (4, \mathbf{v}')   \}$. 
Note that,   within $\imath_u$, states $3$ and $4$ are reached from state $0$ 
in one step and two steps, respectively. 
Therefore, to ensure information consistency, according to Eq.~\eqref{eq:curr-consis},
we have 
\[
\mathbf{v}'[0]
=\mathbf{v}[1]
=(\texttt{N},\texttt{N},\texttt{Y})[2]
= \texttt{Y}.
\]
However, state $4\notin X_C$ is not a critical state. 
According to Eq.~\eqref{eq:one-step-consis}, 
we have $\mathbf{v}'[0]=\texttt{N}$, which is a contradiction. 
Therefore, no such information pattern exists and 
$\mathbb{A}(\{(0,(\texttt{N},\texttt{N},\texttt{Y}))\},\Sigma)=\emptyset$. 
Essentially, this means that 
associating prediction vector $(\texttt{N},\texttt{N},\texttt{Y})$ to the initial state $0$ is meaningless as it cannot be realized in the future.
\end{myexm}

Next, we provide an example of a correct single-observation information pattern that is consistent. 
\begin{myexm}[\bf Consistent Information Patterns]
We still consider the running example shown in Figure~\ref{fig:G} with \( H = 2 \) and $X_C=\{7\}$.
Now, let us first consider a possible information state
$\imath=\{   
(5,(\texttt{N},\texttt{Y},\texttt{N})),
(6,(\texttt{N},\texttt{N},\texttt{N}))
\}$.
This information state could be reached by 
first disabling event $a$ and then observing event $o_1$. 
Suppose the control decision is $\gamma=\Sigma$, which enables all events. 
Then the unique single-observation information pattern consistent with 
$\imath$ under $\gamma$ is 
$\mathcal{I}=(\imath_u,\{\imath_{o_1}\})$, where 
\begin{align}
   \imath_u=&\{(5,(\texttt{N},\texttt{Y},\texttt{N})),
(6,(\texttt{N},\texttt{N},\texttt{N})),
(7,(\texttt{Y},\texttt{N},\texttt{Y}))\} \nonumber\\
\imath_{o_1}=&\{   
(5,(\texttt{N},\texttt{Y},\texttt{N})),
(6,(\texttt{N},\texttt{N},\texttt{N}))
\}.\nonumber
\end{align}
To see the information consistency, 
let us consider $\tilde{x}=(5,(\texttt{N},\texttt{Y},\texttt{N}))\in \imath_u$, and we have
$\mathcal{R}_{\mathcal{I}}(\tilde{x})=\{  (7,(\texttt{Y},\texttt{N},\texttt{Y}))  \}$
as state $7$ can be reached by state $5$ via unobservable event $e$. 
Clearly, it is consistent because
\begin{itemize}
    \item 
    $\mathbf{v}^5[0]\!=\!\texttt{N}$, i.e., the current state is consistent with its current status $5\!\notin\! X_C$;
    \item 
    $\mathbf{v}^5[1]\!=\!\texttt{Y}$ and for the  unique one-step reachable state,  we have $\mathbf{v}^7[0]=\texttt{Y}$, 
    i.e., the one-step prediction is consistent with the actual status in the next step. The same reason for its two-step prediction.
\end{itemize}
Similarly, for the augmented state 
$\tilde{x}'=(6,(\texttt{N},\texttt{N},\texttt{N}))\in \imath_u$, 
we have
$\mathcal{R}_{\mathcal{I}}(\tilde{x}')=\{ \tilde{x}'\}$, 
as it can only reach itself by the observable self-loop $o_1$. 
We can check that this augmented state is also consistent.
\end{myexm}

\section{Supervisor Realizations and  Control Synthesis Algorithms}\label{sec:syhthesis}
In this section, we show how to solve the prediction-based supervisor synthesis problem. 
First, we provide the information-state-based (IS-based) control structure 
and discuss how it can represent a finite-state realizable solution to the problem. 
Then we provide an approach that can effectively search for a valid such structure. 
Finally, we show that the proposed synthesis algorithm is both sound and complete. 
\subsection{IS-Based Control Structures}

Given an information state $\imath$, the supervisor may have multiple control decisions $\gamma \in \Gamma$ at the current instant. Furthermore, under a chosen $\gamma$, there may also exist multiple consistent single-observation information patterns $\mathcal{I}$, i.e., $|\mathbb{A}(\imath, \gamma)| > 1$, since the future evolution of the system remains unresolved at this stage. 
However, when the supervisor’s functionality is fully specified, these choices of $\gamma$ and $\mathcal{I}$ must be uniquely determined, as the future behavior of the closed-loop system is already fixed a priori by the supervisor. 
To formalize this, we introduce the following control structure to capture the unique information-state evolution under a given supervisor.

\begin{mydef}[\bf Control Structures]\upshape\label{def:con-struc}
An information-state-based (IS-based) control structure is a tuple
\begin{equation}
\S=( \mathbb{I}^\S,\mathbb{A}^\S,f^\S_{\mathbb{I},\mathbb{A}},f^\S_{\mathbb{A},\mathbb{I}},\imath_0^\S ),
\end{equation}
where
\begin{itemize}
    \item 
    $\mathbb{I}^\S\subseteq \mathbb{I}$ is a set of information states, which are also referred to as the \emph{decision-states}; 
    \item 
    $\mathbb{A}^\S\subseteq \mathbb{A}$ is a set of consistent single-observation information patterns, which are also referred to as the \emph{observation-states};
    \item 
    $f^\S_{\mathbb{I},\mathbb{A}}: \mathbb{I}^\S\times \Gamma\to\mathbb{A}^\S$  
    is the \emph{deterministic} transition function from decision states to observation states such that, 
    for each $\imath\in \mathbb{I}^\S$, we have
    \begin{itemize}
        \item 
        there exists a unique $\gamma\in \Gamma$ such that $f^\S_{\mathbb{I},\mathbb{A}}(\imath,\gamma)!$; and 
        \item 
        for such  unique $\gamma\in \Gamma$, 
        we have  $f^\S_{\mathbb{I},\mathbb{A}}(\imath,\gamma)\in \mathbb{A}(\imath,\gamma)$. 
    \end{itemize} 
    \item 
    $f^\S_{\mathbb{A},\mathbb{I}}: \mathbb{A}^\S\times \Sigma_o \to\mathbb{I}^\S$  
    is the \emph{deterministic} transition function from observation states to decision states such that, 
    for each $\mathcal{I}=(\imath_u,\{\imath_\sigma\}_{\sigma\in \Sigma_o\cap \gamma})\in \mathbb{A}$, 
    the following transitions are defined
    \begin{equation}
        f^\S_{\mathbb{A},\mathbb{I}}(\mathcal{I},\sigma)= \imath_\sigma,\forall \sigma\in \Sigma_o\cap \gamma
    \end{equation}
    \item 
    $\imath_0^\S=\{ (x_0,\mathbf{v})  \}\in \mathbb{I}^\S$ is the initial state, which is a decision-state
    with a single augmented state such that $x_0$ is the initial state of the system and $\mathbf{v}$ is an arbitrary prediction vector. \hfill$\blacksquare$
\end{itemize}    
\end{mydef}

Essentially, our purpose is to use an IS-based control structure $\S$ as a \emph{finite realization} of a partial-observation supervisor. 
Particularly, at each decision state, since the control decision defined is unique, one can decode this unique transition defined  as the control decision that applies to the system. 
Then we move to track the  successor observation state following the deterministic transition function $f^\S_{\mathbb{I},\mathbb{A}}$. 
At each observation state, all possible observations $\sigma\in \Sigma_o\cap \gamma$ are defined, 
and upon each observation, we move to track the next decision state, from which one can further decode the new control decision, following the deterministic  transition function $f^\S_{\mathbb{A},\mathbb{I}}$. 
Therefore, for each observation sequence
$\alpha=\sigma_1\sigma_2\dots \sigma_n\in \Sigma_o^*$, 
it induces a unique path 
\begin{equation}\label{eq:flow}
    \imath_0\xrightarrow{\gamma_0}\mathcal{I}_0\xrightarrow{\sigma_1}\imath_1\xrightarrow{\gamma_1}\cdots
   \xrightarrow{\gamma_{n-1}}\mathcal{I}_{n-1}\xrightarrow{\sigma_n}\imath_n \xrightarrow{\gamma_{n}}\mathcal{I}_{n},
\end{equation}
where $\gamma_i$ is the unique control decision defined at decision-state $\imath_i$.  
We denote by 
$\mathbb{I}^\S(\alpha)=\imath_n$ and  
$\mathbb{A}^\S(\alpha)=\mathcal{I}_n$ 
the decision-state and the observation-state reached upon $\alpha$ in $\S$, respectively. Based on the above process, we can decode a supervisor from the IS-based control structure as follows.

\begin{mydef}[\bf Induced Supervisors]\upshape
Given an IS-based control structure  $\S$,  
its induced supervisor $S:\Sigma_o^*\to \Gamma$ is defined by: 
\begin{equation}
\forall \alpha=\sigma_1\sigma_2\dots \sigma_n\in \Sigma_o^*: 
S(\alpha)= \gamma_{n}, 
\end{equation}
where $\gamma_n$ is the unique control decision at $\mathbb{I}^\S(\alpha)$ as defined in Eq.~\eqref{eq:flow}.\hfill$\blacksquare$
\end{mydef}

We illustrate the  notions of the control structure and the induced supervisor with the following example.

\begin{myexm}[\bf Control Structure and Its Induced Supervisor]

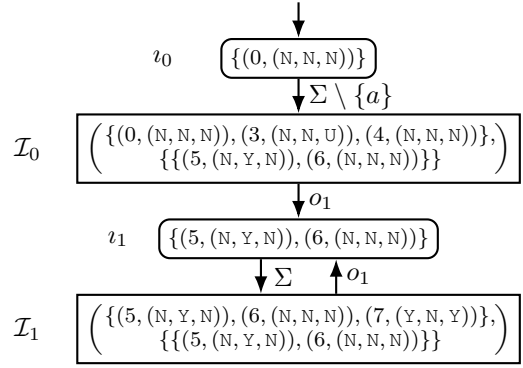
\begin{figure}
    \centering
        \begin{tikzpicture}[->,>={Latex}, thick, initial text={}, node distance= 1.2 cm, initial where=above, thick, Y node/.style={rectangle, rounded corners, draw, minimum size=4mm, font=\footnotesize}, Z node/.style={rectangle, draw, minimum size=4mm, font=\footnotesize}]    
   \node[initial, state, Y node] (0) { $ \{ (0,(\texttt{N}, \texttt{N}, \texttt{N})) \} $ } ;  
   \node[state, Z node, ] [below of = 0] (1) { $\left(\!\!\!\!\!\!\!\!\!\!\begin{array}{cc}
        & \{(0,(\texttt{N}, \texttt{N}, \texttt{N})),(3,(\texttt{N}, \texttt{N}, \texttt{U})), (4,(\texttt{N}, \texttt{N}, \texttt{N}))\}, \\
        &\{\{(5,(\texttt{N}, \texttt{Y}, \texttt{N})),(6,(\texttt{N}, \texttt{N}, \texttt{N}))\}\}
   \end{array}\!\!\!\!\right) $ } ;   
   \node[state, Y node][below of = 1] (2) {$ \{ (5,(\texttt{N}, \texttt{Y}, \texttt{N})), (6,(\texttt{N}, \texttt{N}, \texttt{N})) \} $};
   \node[state, Z node][below of = 2] (3) {$\left(\!\!\!\!\!\!\!\!\!\!\begin{array}{cc}
        & \{(5,(\texttt{N}, \texttt{Y}, \texttt{N})),(6,(\texttt{N}, \texttt{N}, \texttt{N})), (7,(\texttt{Y}, \texttt{N}, \texttt{Y}))\},\\
        &\{\{(5,(\texttt{N}, \texttt{Y}, \texttt{N})),(6,(\texttt{N}, \texttt{N}, \texttt{N}))\}\}
   \end{array}\!\!\!\!\right) $};

   \node[xshift=-1.8cm,]{$\imath_0$};
   \node[xshift=-3.6cm, yshift=-1.2cm]{$\mathcal{I}_0$};
   \node[xshift=-2.4cm, yshift=-2.4cm]{$\imath_1$};
   \node[xshift=-3.6cm, yshift=-3.6cm]{$\mathcal{I}_1$};

   \path[->]
   (0) edge node [xshift=0.7cm] {$\Sigma\setminus \{a\}$} (1)
   (1) edge node [xshift=0.3cm] {$o_1$} (2)
   (2) edge [transform canvas={xshift=-0.5cm}] node [xshift=0.3cm] {$\Sigma$} (3)
   (3) edge [transform canvas={xshift=0.5cm}] node [xshift=0.3cm] {$o_1$} (2);
   \end{tikzpicture}
    \caption{Control Structure of \(S_1\) in Figure~\ref{fig:G1}.}
    \label{fig:controlstructure}
\end{figure}

We still consider the running example shown in Figure~\ref{fig:G} with \(H=2\) and \(X_C=\{7\}\). 
An example of  IS-based control structure is shown in Figure~\ref{fig:controlstructure}. 
In this structure, each rectangle with rounded corners represents a decision-state $\imath \in \mathbb{I}^\S$, from which a unique control decision $\gamma$ is selected; 
each plain rectangle represents an observation-state $\mathcal{I} \in \mathbb{A}^\S$, showing a consistent single-observation information pattern conditioned on the selected control decision. 

We start from the initial decision-state $\imath_0 = \{(0,(\texttt{N},\texttt{N},\texttt{N}))\}$, where the control decision is $\gamma_0 = \Sigma \setminus \{a\}$, i.e., the event \(a\) is disabled. Under this decision, the system evolves to the observation-state $\mathcal{I}_0$, which is a consistent single-observation information pattern of \(\imath_0\) under \(\gamma_0\). 
Then upon observation \(o_1\), the structure transitions to decision-state $\imath_1 = \{(5,(\texttt{N}, \texttt{Y}, \texttt{N})), (6,(\texttt{N}, \texttt{N}, \texttt{N}))\}$, at which the unique control decision $\gamma_1 = \Sigma$ (all events enabled) is applied, leading deterministically to the next observation-state $\mathcal{I}_1$. Finally, the structure loops via observation \(o_1\) and transitions back to $\imath_1$.
This control structure induces   a unique supervisor \(S_1\) that disables event \(a\) at the beginning and enables all events afterward, whose controlled behavior is already as shown in Figure~\ref{fig:G1}.
\end{myexm}

Note that in the IS-based control structure, each decision-state consists of a set of system states augmented with prediction vectors. 
In fact, if the supervisor always selects control decisions according to the control structure (as is the case for the induced supervisor), 
then for the  decision-state reached in the IS-based control structure $\S$, we have:
\begin{itemize}
\item[(i)] 
Its system state component  corresponds to the current state estimate (excluding unobservable tails) under the induced supervisor $S$; and
\item[(ii)] 
Its  prediction vector component matches the prediction set of the induced supervisor $S$.
\end{itemize}

For example, let us consider the decision-state
$\imath_1 = \{(5,(\texttt{N}, \texttt{Y}, \texttt{N})), (6,(\texttt{N}, \texttt{N}, \texttt{N}))\}$ as shown in Figure~\ref{fig:controlstructure}.
This decision-state is reached after observations $\alpha = o^*_1$, and we have \(\mathcal{O}(\alpha)=\L(S_1/G)\cap\Sigma^*\{o_1\}\). Then, the state component \(\mathsf{state}(\imath_1)=\{5,6\}=\{\delta(s)\in X: s\in \O(\alpha)\}\}\) is exactly the current state estimate (excluding unobservable tails) based on $\alpha$ and the control decisions made by \(S_1\) along the path; and \(\mathsf{vec}(\imath_1)=\{(\texttt{N}, \texttt{Y}, \texttt{N}),(\texttt{N}, \texttt{N}, \texttt{N})\}=\{\xi^S(s)\in\mathbb{V}:s\in \O(\alpha)\}\) is exactly the prediction set for \(\alpha\) under \(S_1\).
This relationship between the control structure and its induced supervisor is formally established as follows. 

\begin{mypro}\label{prop:IS}\upshape
Let $\S$ be an IS-based control structure and $S:\Sigma_o^*\to \Gamma$ be its induced partial-observation supervisor.  
For each observation $\alpha\in P(\mathcal{L}(S/G))$,  we have
\begin{align}
    \state(\mathbb{I}^\S(\alpha))=&\ \{\delta(s)\in X: s\in \mathcal{O}(\alpha)\}, \\
    \mathsf{vec}(\mathbb{I}^\S(\alpha))=&\ \Xi^S(\alpha). 
\end{align} 
\end{mypro}
\begin{proof}
The proof is provided in the Appendix.
\end{proof}

\subsection{Supervisor Synthesis Algorithm}

Now, we tackle the supervisor synthesis problem. 
Recall that, in our problem, we need to ensure both the liveness of the closed-loop system  and the satisfaction of the prediction-based property $(X_C,\Phi)$. 
These two requirements can both be evaluated based on the information state as follows. 

\begin{mydef}[\bf Liveness and Safety]\upshape
Let $\S$ be an IS-based control structure and 
$\imath\in \mathbb{I}^\S$ be a decision-state with $\gamma\in \Gamma$ be the unique  decision defined at $\imath$. 
We say decision-state $\imath$ is 
\begin{itemize}
    \item 
    \textbf{live} if any state within its unobservable reach cannot be blocked by the control decision, i.e., 
    \begin{equation}
        \forall x\in \UR_\gamma(  \state(\imath) ),\exists \sigma\in \gamma: \delta(x,\sigma)!.
    \end{equation}
    \item 
    \textbf{safe} if its prediction vector component satisfies  $(X_C,\Phi)$, i.e., $ \Phi( \mathsf{vec}(\imath) )=1$.\hfill$\blacksquare$
\end{itemize}
\end{mydef}

The following theorem states that, in order to synthesize a live supervisor that enforces the prediction-based property, it suffices to find an IS-based control structure in which all states are live and safe. 
\begin{mythm}\upshape\label{thm1}
Let $\S$ be an IS-based control structure 
such that all decision-states in it are live and safe. 
Then its induced partial-observation supervisor $S:\Sigma_o^*\to \Gamma$   is a solution to Problem~1, 
i.e., $\mathcal{L}(S/G)$ is live and  $S/G\models (X_C,\Phi)$.    
\end{mythm}
\begin{proof}
The proof is provided in the Appendix.
\end{proof}

In general, finding a control structure $\S$ satisfying Theorem~1 is a challenging task because there are choices whose consequences cannot be evaluated immediately. 
Nevertheless, since the solution space is bounded, 
one can  first enumerate all possible configurations and then extract a feasible IS-based structure, in which all states are live and safe,  from the bounded space. 
Such an idea is implemented by Algorithm~\ref{alg:main}, which consists of the following three steps.

\textbf{Step 1--Expand the Solution Space (line 1-5)}: 
The objective of this step is to initially build a structure 
\[
\B=( \mathbb{I}^\B,\mathbb{A}^\B,f^\B_{\mathbb{I},\mathbb{A}},f^\B_{\mathbb{A},\mathbb{I}}, \mathbb{I}_0^\B )
\]
that enumerates all possible information-states that are live, safe, and consistent.
This structure employs a finite intermediate structure
 to facilitate the search for a feasible $\S$. 
Compared to the control structure defined in Definition~\ref{def:con-struc}, $\B$ differs in the following two aspects:
\begin{itemize} 
    \item[1)] 
    Transition function
    $f^\B_{\mathbb{I},\mathbb{A}}\subseteq  \mathbb{I}^\B\times \Gamma\times\mathbb{A}^\B$   
    is allowed to be non-deterministic; and 
    \item[2)] 
    The initial state set $\mathbb{I}_0^\B \subseteq  \mathbb{I}^\B$ is not required to be a singleton.
\end{itemize} 
More specifically, 
in lines~1-2, we  set \(\B\)  to contain only safe and consistent initial states of the form
\(\{(x_0,\mathbf{v})\}\), 
where   
safety is ensured by   \(\Phi(\{\mathbf{v}\})=1\), 
and  consistency with the current state is ensured by
\(\mathbf{v}[0]=\chi_C(\{x_0\})\). 
Starting from each such initial state, we invoke the procedure \texttt{Expand}, which is a recursive procedure, to iteratively grow  \(\B\) via depth-first search, exploring only live, safe, and consistent information states. The search terminates upon encountering either:
(i) an information state violating liveness, safety, or consistency; or
(ii) a  previously visited information state. 
By construction, all states in \(\mathbb{I}^{\B}\) and \(\mathbb{A}^{\B}\) are guaranteed to satisfy liveness, safety, and information consistency. 

\begin{myexm}[\bf Initial Expansion]

\begin{figure*}
    \centering
        \begin{tikzpicture}[->,>={Latex}, thick, initial text={}, node distance= 1.2 cm, initial where=above, thick, Y node/.style={rectangle, rounded corners, draw, minimum size=4mm, font=\footnotesize}, Z node/.style={rectangle, draw, minimum size=4mm, font=\footnotesize}]  
   \node[initial, state, Y node] [xshift=0cm, yshift=1.8cm] (0) { $ \{ (0,(\texttt{N}, \texttt{N}, \texttt{N})) \} $ } ;  
   \node[state, Z node, color=red] [xshift=-12cm, yshift=1.5cm] (1) { $\left(\!\!\!\!\!\!\!\!\!\!\begin{array}{cc}
        & \left\{\!\!\!\!\!\!\!\!\!\!\begin{array}{cc}
             & (0,(\texttt{N}, \texttt{N}, \texttt{N})),(1,(\texttt{N}, \texttt{N}, \texttt{N})),\\
        & (2,(\texttt{N}, \texttt{N}, \texttt{U})),(3,(\texttt{N}, \texttt{N}, \texttt{U})),\\
        &(4,(\texttt{N}, \texttt{N}, \texttt{N}))
        \end{array} \!\!\!\!\right\}, \\
        & \left\{ \!\!\!\!\!\!\!\!\!\!\begin{array}{cc}
             &  \{(2,(\texttt{N}, \texttt{N}, \texttt{U})),(5,(\texttt{N}, \texttt{Y}, \texttt{N})),\\ 
             &(6,(\texttt{N}, \texttt{N}, \texttt{N}))\},\{(5,(\texttt{N}, \texttt{Y}, \texttt{N}))\}
        \end{array}\!\!\!\!
        \right\}
   \end{array}\!\!\!\!\right)$ } ;   

   \node[state, Z node][xshift=0cm, yshift=-5.5cm] (2) { $\left(\!\!\!\!\!\!\!\!\!\!\begin{array}{cc}
        & \left\{\!\!\!\!\!\!\!\!\!\!\begin{array}{cc}
             &  (0,(\texttt{N}, \texttt{N}, \texttt{N})),(3,(\texttt{N}, \texttt{N}, \texttt{U})),\\
             & (4,(\texttt{N}, \texttt{N}, \texttt{Y}))
        \end{array}\!\!\!\!\right\}, \\
        &\{\{(5,(\texttt{N}, \texttt{Y}, \texttt{N})),(6,(\texttt{N}, \texttt{Y}, \texttt{N}))\}\}
   \end{array}\!\!\!\!\right) $  } ; 

   \node[state, Y node, color=red][xshift=-6.5cm, yshift=-5.5cm] (3) { $ \{ (5,(\texttt{N}, \texttt{Y}, \texttt{N})), (6,(\texttt{N}, \texttt{Y}, \texttt{N})) \} $ } ; 

   \node[state, Y node, ] [xshift=-12cm, yshift=-0.2cm] (4) { $\{(2,(\texttt{N}, \texttt{N}, \texttt{U})), (5,(\texttt{N}, \texttt{Y}, \texttt{N})), (6,(\texttt{N}, \texttt{N}, \texttt{N}))\} $ } ;   
   
   \node[state, Z node, color=red] [xshift=-12cm, yshift=-1.7cm] (5) { 
   $\left(\!\!\!\!\!\!\!\!\!\!\begin{array}{cc}
        & \left\{\!\!\!\!\!\!\!\!\!\!\begin{array}{cc}
             & (2,(\texttt{N}, \texttt{N}, \texttt{U})),(5,(\texttt{N}, \texttt{Y}, \texttt{N}))\\
        & (6,(\texttt{N}, \texttt{N}, \texttt{N})),(7,(\texttt{Y}, \texttt{N}, \texttt{Y}))
        \end{array} \!\!\!\!\right\}, \\
        & \left\{ \!\!\!\!\!\!\!\!\!\!\begin{array}{cc}
             &  \{(2,(\texttt{N}, \texttt{N}, \texttt{U})),(5,(\texttt{N}, \texttt{Y}, \texttt{N})), \\
             &(6,(\texttt{N}, \texttt{N}, \texttt{N}))\},\{(5,(\texttt{N}, \texttt{Y}, \texttt{N}))\}
        \end{array}\!\!\!\!
        \right\}
   \end{array}\!\!\!\!\right)$ } ; 
   
   \node[state, Z node, color=red] [xshift=-6.5cm, yshift=0.7cm] (6) { $\left(\!\!\!\!\!\!\!\!\!\!\begin{array}{cc}
        & \{(0,(\texttt{N}, \texttt{N}, \texttt{N})),(1,(\texttt{N}, \texttt{N}, \texttt{N})), \\
        & (3,(\texttt{N}, \texttt{N}, \texttt{Y}))\},\{\{(5,(\texttt{N}, \texttt{Y}, \texttt{N}))\}\}
   \end{array}\!\!\!\!\right) $ } ; 
   \node[state, Z node, color=red] [xshift=-6.5cm, yshift=-0.8cm] (7) { $\left(\!\!\!\!\!\!\!\!\!\!\begin{array}{cc}
        & \{(0,(\texttt{N}, \texttt{N}, \texttt{N})),(3,(\texttt{N}, \texttt{N}, \texttt{Y}))\}, \\
        &\{\{(5,(\texttt{N}, \texttt{Y}, \texttt{N}))\}\}
   \end{array}\!\!\!\!\right) $ } ; 
   
   \node[state, Z node, color=red] [xshift=-6.5cm, yshift=-2.4cm] (8) { $\left(\!\!\!\!\!\!\!\!\!\!\begin{array}{cc}
        & \left\{\!\!\!\!\!\!\!\!\!\!\begin{array}{cc}
             & (0,(\texttt{N}, \texttt{N}, \texttt{N})),(1,(\texttt{N}, \texttt{N}, \texttt{N})), \\
             & (2,(\texttt{N}, \texttt{N}, \texttt{U}))
        \end{array}\!\!\!\!\right\},\\
        &\{\{(2,(\texttt{N}, \texttt{N}, \texttt{U}))\},\{(5,(\texttt{N}, \texttt{Y}, \texttt{N}))\}\}
   \end{array}\!\!\!\!\right)$ } ;
   \node[state, Y node][xshift=-6.5cm, yshift=-3.8cm] (9) {$ \{ (2,(\texttt{N}, \texttt{N}, \texttt{U})) \} $};
   \node[state, Z node, color=red][xshift=-12cm, yshift=-3.8cm] (10) {$\left(\!\!\!\!\!\!\!\!\!\!\begin{array}{cc}
        & \{(2,(\texttt{N}, \texttt{N}, \texttt{U}))\},\\
        &\{\{(2,(\texttt{N}, \texttt{N}, \texttt{U}))\},\{(5,(\texttt{N}, \texttt{Y}, \texttt{N}))\}\}
   \end{array}\!\!\!\!\right)$};
   \node[state, Z node, ] [xshift=0cm,yshift=-0.9cm] (11) { $\left(\!\!\!\!\!\!\!\!\!\!\begin{array}{cc}
        & \left\{\!\!\!\!\!\!\!\!\!\!\begin{array}{cc}
             &  (0,(\texttt{N}, \texttt{N}, \texttt{N})),(3,(\texttt{N}, \texttt{N}, \texttt{U})),\\
             & (4,(\texttt{N}, \texttt{N}, \texttt{N}))
        \end{array}\!\!\!\!\right\}, \\
        &\{\{(5,(\texttt{N}, \texttt{Y}, \texttt{N})),(6,(\texttt{N}, \texttt{N}, \texttt{N}))\}\}
   \end{array}\!\!\!\!\right) $ } ; 
   \node[state, Y node, ] [xshift=0cm, yshift=-2.2cm] (12) { $ \{ (5,(\texttt{N}, \texttt{Y}, \texttt{N})), (6,(\texttt{N}, \texttt{N}, \texttt{N})) \} $ } ; 
   \node[state, Z node, ] [xshift=0cm, yshift=-3.5cm] (13) { $\left(\!\!\!\!\!\!\!\!\!\!\begin{array}{cc}
        & \left\{\!\!\!\!\!\!\!\!\!\!\begin{array}{cc}
             &  (5,(\texttt{N}, \texttt{Y}, \texttt{N})),(6,(\texttt{N}, \texttt{N}, \texttt{N})),\\
             & (7,(\texttt{Y}, \texttt{N}, \texttt{Y}))
        \end{array}\!\!\!\!\right\},\\
        &\{\{(5,(\texttt{N}, \texttt{Y}, \texttt{N})),(6,(\texttt{N}, \texttt{N}, \texttt{N}))\}\}
   \end{array}\!\!\!\!\right) $ } ; 
   \node[state, Z node, ] [xshift=0cm, yshift=0.6cm] (14) { $\left(\left\{\!\!\!\!\!\!\!\!\!\!\begin{array}{cc}
        &  (0,(\texttt{N}, \texttt{Y}, \texttt{N}))),\\
        & (1,(\texttt{N}, \texttt{Y}, \texttt{N})))
   \end{array}\!\!\!\!\right\}\right)$ } ; 

   \node[initial, state, Y node, color=red] [xshift=-12cm, yshift=-5.5cm] (15) { $ \{ (0,(\texttt{N}, \texttt{N}, \texttt{U})) \} $ } ;

   \draw[->] (0) -- node [xshift=0.1cm, yshift=0.2cm] {$\Sigma\setminus\{a\}$} (2cm, 1.8cm)--([xshift=2cm]11.north);
   \draw[->] (0) -- (2.8cm, 1.8cm)--(2.8cm, -5.5cm)--(2);
   \draw[-, draw = blue] (-2.5cm,2.6cm) -- (2.6cm,2.6cm) -- (2.6cm, -4.5cm) -- (-2.5cm, -- -4.5cm)-- (-2.5cm, 2.6cm);
 
   \path[->]
   (0) edge node [xshift=1.8cm, yshift=0.2cm] {$\Sigma$} ([yshift=0.3cm]1.east)
   (0) edge node [xshift=0.8cm] {$\Sigma\setminus\{b,c\}$} (14)
   (1) edge node [xshift=0.3cm] {$o_1$} (4)
   (2) edge node [yshift=0.2cm] {$o_1$} (3)
   (0.west) edge node [xshift=-0.8cm, yshift=0.1cm] {$\Sigma\setminus\{b\}$} (6.east)
   (0.west) edge node [xshift=-1.4cm, yshift=-0.5cm] {$\Sigma\setminus\{a,b\}$} (7.east)
   (0.west) edge node [xshift=-1.2cm, yshift=-0.8cm] {$\Sigma\setminus\{c\}$} (8.east)

   (11) edge node [xshift=0.3cm] {$o_1$} (12)
   (4) edge [transform canvas={xshift=0.3cm}] node [xshift=0.3cm] {$\Sigma$} (5)
   (5) edge [transform canvas={xshift=-0.3cm}] node [xshift=-0.3cm] {$o_1$} (4)
   (8) edge node [xshift=0.3cm] {$o_1$} (9)
   (9) edge [transform canvas={yshift=0.1cm}] node [yshift=0.3cm] {$\Sigma$} (10)
   (10) edge [transform canvas={yshift=-0.1cm}] node [yshift=-0.3cm] {$o_1$} (9)
   (12) edge [transform canvas={xshift=0.3cm}] node [xshift=0.3cm] {$\Sigma$} (13)
   (13) edge [transform canvas={xshift=-0.3cm}] node [xshift=-0.3cm] {$o_1$} (12)
   ;
   \end{tikzpicture}
    \caption{Partial representation of structure \(\B\). Incomplete states are highlighted in red. The structure with all states complete is in the blue-lined box.}
    \label{fig:alg}
\end{figure*}
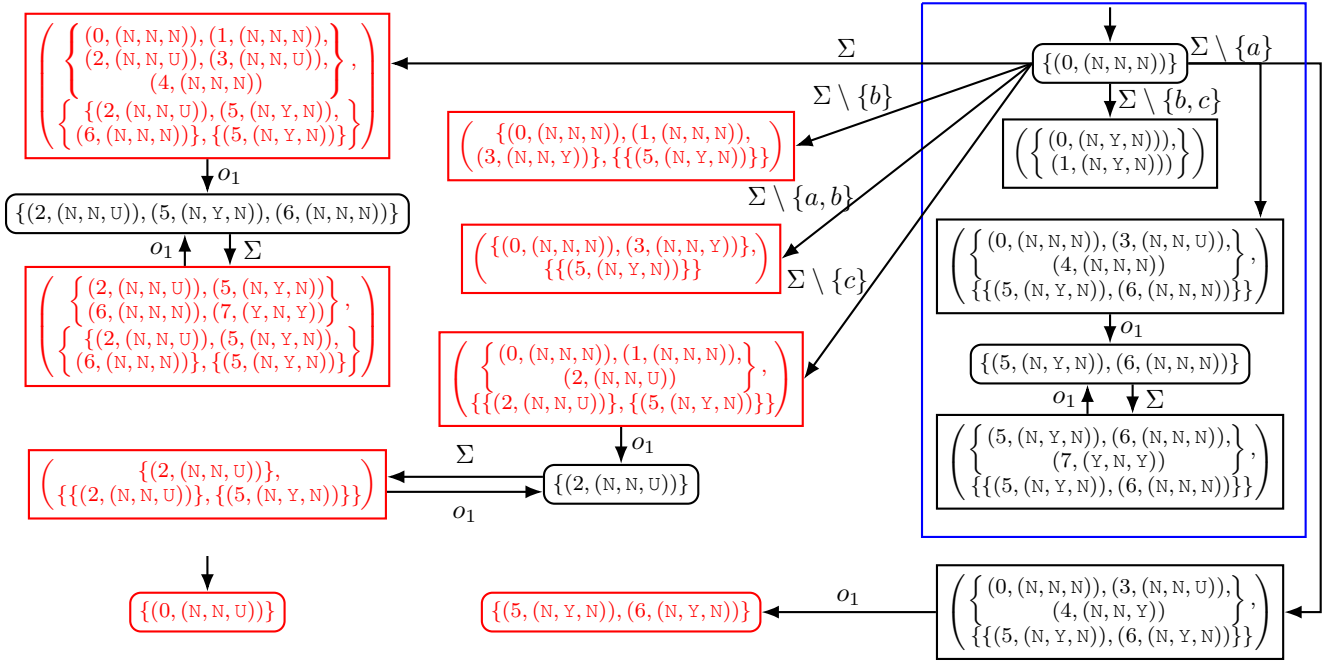

Let us continue with the running example. A portion of the overall structure \(\B\) is illustrated in Figure~\ref{fig:alg}. Initially, we set the initial state set as \(\mathbb{I}^\B_0 = \{\{(0,(\texttt{N}, \texttt{N}, \texttt{N}))\}, \{(0,(\texttt{N}, \texttt{N}, \texttt{U}))\}, \ldots\}\). Since \(\mathbb{A}(\{(0,(\texttt{N}, \texttt{N}, \texttt{U}))\}, \gamma) = \emptyset\) for all control decisions \(\gamma \in \Gamma\), the search terminates at this initial state. 
Starting from  another possible initial state \(\imath_0 = \{(0,(\texttt{N}, \texttt{N}, \texttt{N}))\}\), the search iteratively expands \(\B\) along decisions that do not disable both events \(a\) and \(c\) simultaneously, which preserve the liveness of \(\imath_0\), until it encounters the unsafe decision state \(\{(5,(\texttt{N}, \texttt{Y}, \texttt{N}))\}\).
It is worth noting that the same control decision, such as \(\Sigma\setminus\{a\}\), may lead \(\B\) to different consistent observation-states. However, some of them will lead to inconsistent states in the future. For instance, consider the state
\[
\left(\!\!\!\!\!\!\!\!\!\!\begin{array}{cc}
        & \left\{\!\!\!\!\!\!\!\!\!\!\begin{array}{cc}
             &  (0,(\texttt{N}, \texttt{N}, \texttt{N})),(3,(\texttt{N}, \texttt{N}, \texttt{U})),\\
             & (4,(\texttt{N}, \texttt{N}, \texttt{Y}))
        \end{array}\!\!\!\!\right\}, \\
        &\{\{(5,(\texttt{N}, \texttt{Y}, \texttt{N})),(6,(\texttt{N}, \texttt{Y}, \texttt{N}))\}\}
   \end{array}\!\!\!\!\right),
\] 
it can be later verified that its successor decision state \(\{(5,(\texttt{N}, \texttt{Y}, \texttt{N})),(6,(\texttt{N}, \texttt{Y}, \texttt{N}))\}\) is inconsistent, and thus the search terminates at this point.
\end{myexm}

\begin{algorithm}[t]
\caption{Synthesis of Control Structure \(\S\)}\label{alg:main} 
\KwIn{\(G\), \(\Sigma_c\), \(\Sigma_o\), \((X_C,\Phi)\)}
\KwOut{\(\S\) or its induced supervisor}

$\mathbb{I}_0^\B\gets \{\{(x_0,\mathbf{v})\}\in\mathbb{I}: \mathbf{v}[0]=\chi_C(\{x_0\}), \Phi(\{\mathbf{v}\})=1\} $
\\
$\mathbb{I}^\B\gets\mathbb{I}_0^\B$,
\ \
$\mathbb{A}^\B\gets\emptyset$,
\ \
$f^\B_{\mathbb{I},\mathbb{A}}\gets \emptyset$, 
\ \ 
$f^\B_{\mathbb{A},\mathbb{I}}\gets \emptyset$
\\
\For{$\imath_0\in \mathbb{I}_0^\B$}
{
$\texttt{Expand}(\B,\imath_0)$
} 
$\texttt{Prune}(\B)$
\\
\eIf{$\mathbb{I}_0^\B=\emptyset$}
{
\textbf{return} no solution exists}
{
Pick $\imath_0 \in \text{max}_\texttt{U}(\mathbb{I}^{\B}_0)$
\\
\(\mathbb{I}^{\S} \leftarrow \{ \imath_0 \}\), \(\mathbb{A}^{\S} \leftarrow \emptyset\), \(f^{\S} \leftarrow\emptyset\) \\
\(\S\gets \texttt{Extract}(\S,B,\imath_0)\)\\
\textbf{return} IS-based Control Structure \(\S\)} 
\end{algorithm}

\begin{algorithm}[t]
\renewcommand{\algorithmcfname}{Procedure}
\renewcommand{\thealgocf}{}                  
\caption{\(\texttt{Expand}(\B,\imath)\)}    
\addtocounter{algocf}{-1}   
\For{\(\gamma\in\Gamma\) \emph{such that} \(\imath\) \emph{is live under} \(\gamma\)}
{
        \For{\(\mathcal{I}=(\imath_u,\{\imath_\sigma\}_{\sigma\in \Sigma_o\cap \gamma})\in\mathbb{A}(\imath,\gamma)\)}
        {
            \(f^\B_{\mathbb{I},\mathbb{A}} \leftarrow f^\B_{\mathbb{I},\mathbb{A}} \cup \{ (\imath,\gamma,\mathcal{I}) \}\)\\
            \If{\(\mathcal{I}\not\in \mathbb{A}^{\B}\)}
            {
                \(\mathbb{A}^{\B} \leftarrow \mathbb{A}^{\B} \cup \{ \mathcal{I} \}\)\\
                    \For{\(\sigma\in\gamma\cap\Sigma_o\)}
                    {
                        \If{\(\imath_\sigma\) is safe}
                        {
                            \(f^\B_{\mathbb{A},\mathbb{I}} \leftarrow f^\B_{\mathbb{A},\mathbb{I}} \cup \{ (\mathcal{I},\sigma,\imath_\sigma) \}\)\\
                            \If{\(\imath_\sigma\not\in \mathbb{I}^{\B}\)}
                            {
                                \(\mathbb{I}^{\B}\leftarrow \mathbb{I}^{\B}\cup\{\imath_\sigma\}\)\\
                                \(\texttt{Expand}(\B,\imath_\sigma)\)
                            }
                        }
                    }
                
            }
        }

}
\end{algorithm}

\textbf{Step 2--Prune Incomplete States (line 6)}: 
Note that the structure $\B$ obtained after procedure \texttt{Expand} cannot be directly used for the purpose of control synthesis due to the presence of \emph{incomplete} states. 
Specifically, 
\begin{itemize}
    \item 
    A decision-state $\imath\in \mathbb{I}^\B$ is said to be incomplete 
    if no feasible control decision is defined. 
    Therefore, if such a state is reached, no future decision can be taken in order to ensure liveness, safety and consistency. 
    We define 
    \begin{equation}
  \mathbb{I}_{bad}^\B = \{\imath\in \mathbb{I}^\B \mid (\{\imath\}\times\Gamma\times \mathbb{A}^{\B})\cap f^\B=\emptyset\}
    \end{equation}
    as the set of incomplete decision states in $\B$. 
      \item 
    An observation-state $\mathcal{I}\in \mathbb{A}^\B$ is said to be incomplete if it lacks transitions \emph{for some} feasible observations. When such a state is reached, the occurrence of any missing observation (which is inherently uncontrollable, as the control decision has already been fixed) will force the system into a state that violates liveness, safety, or consistency.  
    We define 
    \begin{equation}
    \mathbb{A}_{bad}^\B\!=\! \{\mathcal{I}\!=\!(\imath_u,\{\imath_\sigma\}_{\sigma\in \Sigma_o\cap \gamma})\!\in\!  \mathbb{A}^\B\mid \exists \imath_\sigma.(\mathcal{I},\sigma,\imath_\sigma)\!\not\in\! f^\B\}  
    \end{equation}
    as the set of incomplete observation states in $\B$. 
\end{itemize}
Therefore, the objective of procedure \texttt{Prune} is to ensure the completeness of $\B$. 
However, removing currently incomplete states may introduce new incompleteness. For instance, eliminating a decision state $\imath_{bad}\in \mathbb{I}^\B_{bad}$ would cause any predecessor observation state $\mathcal{I}$ such that 
$(\mathcal{I},\sigma,\imath_{bad})$
to become incomplete in the modified structure. Consequently, an iterative removal process in the while-loop is required until no incomplete states remain.
 
\begin{myexm}[\bf Iterative Completeness Check]
  We continue with the running example. In Figure~3, all of the incomplete states are highlighted in red. For example, state \(\{ (5,(\texttt{N}, \texttt{Y}, \texttt{N})), (6,(\texttt{N}, \texttt{Y}, \texttt{N})) \}\)  is incomplete because it has no feasible control decision defined. Another example is the state
  \[
  \left(\!\!\!\!\!\!\!\!\!\!\begin{array}{cc}
        & \left\{\!\!\!\!\!\!\!\!\!\!\begin{array}{cc}
             & (0,(\texttt{N}, \texttt{N}, \texttt{N})),(1,(\texttt{N}, \texttt{N}, \texttt{N})), \\
             & (2,(\texttt{N}, \texttt{N}, \texttt{U}))
        \end{array}\!\!\!\!\right\},\\
        &\{\{(2,(\texttt{N}, \texttt{N}, \texttt{U}))\},\{(5,(\texttt{N}, \texttt{Y}, \texttt{N}))\}\}
   \end{array}\!\!\!\!\right),
  \]
  which is incomplete due to the absence of a transition for observation \(o_2\), which will lead \(\B\) to the unsafe state \(\{(5,(\texttt{N}, \texttt{Y}, \texttt{N}))\}\).
  By removing such red-highlighted states from \(\B\), the state
    \[
    \left(\!\!\!\!\!\!\!\!\!\!\begin{array}{cc}
            & \left\{\!\!\!\!\!\!\!\!\!\!\begin{array}{cc}
                 &  (0,(\texttt{N}, \texttt{N}, \texttt{N})),(3,(\texttt{N}, \texttt{N}, \texttt{U})),\\
                 & (4,(\texttt{N}, \texttt{N}, \texttt{Y}))
            \end{array}\!\!\!\!\right\}, \\
            &\{\{(5,(\texttt{N}, \texttt{Y}, \texttt{N})),(6,(\texttt{N}, \texttt{Y}, \texttt{N}))\}\}
       \end{array}\!\!\!\!\right),
    \] 
  also becomes incomplete, and will therefore be removed in a subsequent iteration. After this iterative pruning process, we obtain the complete structure, which is enclosed within the blue-lined box.
\end{myexm}

\begin{algorithm}[t]
\renewcommand{\algorithmcfname}{Procedure}
\renewcommand{\thealgocf}{}                  
\caption{\(\texttt{Prune}(\B)\)}    
\addtocounter{algocf}{-1}   
\While{$\mathbb{I}_{bad}^\B\neq\emptyset$ or $\mathbb{A}_{bad}^\B\neq \emptyset$}
{
Remove all incomplete states $\mathbb{I}_{bad}^\B$ and $\mathbb{A}_{bad}^\B$ from structure $\B$\\
Remove all transitions involving incomplete states $\mathbb{I}_{bad}^\B$ or $\mathbb{A}_{bad}^\B$ from structure $\B$ 
}
\end{algorithm}

\begin{algorithm}[t]
\renewcommand{\algorithmcfname}{Procedure}
\renewcommand{\thealgocf}{}                  
\caption{\(\texttt{Extract}(\S,\B,\imath)\)}    
\addtocounter{algocf}{-1}   
Pick $(\imath,\gamma^\star,\mathcal{I}^\star)\!\in\! f^\B$ such that 
\[
\forall (\imath,\gamma',\mathcal{I}')\!\in\! f^\B:
[\gamma^\star\not\subset\gamma']
\wedge
[(\gamma'=\gamma^\star)\!\Rightarrow\! (\mathcal{I}^\star\not<_\texttt{U}\mathcal{I}')]
\]\\
\(f^{\S} \leftarrow f^{\S} \cup \{ (\imath,\gamma^\star,\mathcal{I}^\star) \}\)\\
\If{\(\mathcal{I}^\star=((\imath_u,\{\imath_\sigma\}_{\sigma\in \Sigma_o\cap \gamma^\star}))\not\in \mathbb{A}^\S\)}
{
\(\mathbb{A}^{\S} \leftarrow \mathbb{A}^{\S} \cup \{ \mathcal{I}^\star \}\) \

\For{\(\sigma \in \gamma^\star \cap \Sigma_o\)}
{
\(f^{\S} \leftarrow f^{\S} \cup \{ (\mathcal{I}^\star,\sigma,\imath_\sigma) \}\)\\
\If{\(\imath_\sigma \notin \mathbb{I}^{\S}\)}
{
\(\mathbb{I}^{\S} \leftarrow \mathbb{I}^{\S} \cup \{ \imath_\sigma \}\)\\
\(\texttt{Extract}(\S,\B,\imath_\sigma)\)
}
}

}
\end{algorithm}

\textbf{Step 3--Extract a Supervisor (line 7-14)}:
The structure $\B$ remaining after procedure \texttt{Prune} can be directly used for control synthesis through the following steps:
\begin{itemize}
    \item[1)] 
    First, we select an initial decision-state $\imath\in \mathbb{I}_0^\B$; 
    \item[2)] 
    From the chosen decision-state $\imath$, we select a control decision $\gamma$ and observation-state $\mathcal{I}$ such that $(\imath, \gamma, \mathcal{I}) \in f_{\mathbb{I},\mathbb{A}}^\B$. Such a choice is guaranteed to exist since $\B$ is now complete;
    \item[3)] 
   For the reached observation-state $\mathcal{I} = (\imath_u, \{\imath_\sigma\}_{\sigma \in \Sigma_o \cap \gamma})$, we consider all successor transitions $(\mathcal{I}, \sigma, \imath_\sigma)$ for each $\sigma \in \Sigma_o \cap \gamma$. These transitions are all well-defined in $\mathcal{B}$ due to the completeness of observation-states;
    \item[4)] 
    We repeat this process via depth-first search until no new states can be visited. This procedure effectively constructs a control structure $\S$ from $\B$, as both the initial state and transition functions are now deterministic.
\end{itemize}
In fact, if one is only interested in enforcing the prediction-based property $(X_C, \Phi)$, the above process already suffices.
Here, we proceed further and seek the ``optimal" selection when multiple choices exist for initial states and control-observation-state pairs $(\gamma, \mathcal{I})$.
To this end, we introduce a partial order as follows.
For two information states $\imath,\imath'\in \mathbb{I}$ such that $\mathsf{state}(\imath)=\mathsf{state}(\imath')$, we define:
    \begin{itemize}
        \item 
        $\imath\leq_\texttt{U} \imath'$  if for all $ (x,\mathbf{v})\in\imath, (x,\mathbf{v}')\in\imath'$ and each instant $i\in\{1,...,H\}$, it holds that 
        $(\mathbf{v}[i]\!=\!\texttt{U}) \Rightarrow (\mathbf{v}'[i]\!=\!\texttt{U})$;
        \item \(\imath<_\texttt{U} \imath'\)  if \(\imath\leq_\texttt{U} \imath'\) and there exists \((x,\mathbf{v})\in\imath, (x,\mathbf{v}')\in\imath', i\in\{1,...,H\}\)  such that
        \((\mathbf{v}[i]\ne\texttt{U})\wedge (\mathbf{v}'[i]=\texttt{U})\). 
    \end{itemize}
When pick an initial state from $\mathbb{I}_0^\B$ in line~10 of Algorithm~\ref{alg:main},
we select a maximal element in $\mathbb{I}_0^\B$, i.e., 
\begin{equation}
\text{max}_{\texttt{U}}(\mathbb{I}_0^\B)
=
\{
\imath \in \mathbb{I}^{\B}_0\mid 
\forall \imath'\in \mathbb{I}^{\B}_0:\imath\not<_\texttt{U} \imath'
\}.
\end{equation}
Given  decision-state \(\imath\) and control decision \(\gamma\), 
for two single-observation information patterns \(\mathcal{I}=(\imath_u,\{\imath_\sigma\}_{\sigma\in \Sigma_o\cap \gamma}),\mathcal{I}'=(\imath'_u,\{\imath'_\sigma\}_{\sigma\in \Sigma_o\cap \gamma})\in\mathbb{A}(\imath,\gamma)\), we define:
    \begin{itemize}
        \item 
        \(\mathcal{I}\leq_\texttt{U}\mathcal{I}'\) if 
        \(\imath_u \leq_\texttt{U} \imath'_u\) and 
        \(\forall \sigma\in \Sigma_o\cap \gamma: \imath_\sigma \leq_\texttt{U} \imath'_\sigma\);
        \item
        \(\mathcal{I}<_\texttt{U}\mathcal{I}'\)  if 
        \(\mathcal{I}\leq_\texttt{U}\mathcal{I}'\) and either \(\imath_u <_\texttt{U} \imath'_u\) or \(\exists (\sigma\in \Sigma_o\cap \gamma: \imath_\sigma <_\texttt{U} \imath'_\sigma\).
    \end{itemize}
Therefore, when selecting the ``optimal" pair $(\gamma^\star, \mathcal{I}^\star)$ in line 1 of procedure \texttt{Extract}, we need to ensure that:
(i) the control decision $\gamma^\star$ is locally maximal, enabling as many events as possible; and
(ii) when control decisions are equivalent, the observation-state is chosen as a maximal element under the partial order $<_\texttt{U}$.
Later, we will prove that such choices not only guarantee property enforcement but also ensure maximal permissiveness.

\begin{myexm}[\bf Synthesize a Control Structure]
  After applying procedures \texttt{Expand} and \texttt{Prune}, we obtain the structure enclosed in the blue-lined box. From this structure, we identify the sole initial decision-state \(\{ (0,(\texttt{N}, \texttt{N}, \texttt{N})) \}\), from which two control decisions \(\Sigma\setminus\{a\}\) and \(\Sigma\setminus\{b,c\}\) are available, each leading to a different observation-state. If we choose \(\Sigma\setminus\{b,c\}\), the resulting control structure consists of the two states in the top-right corner of Figure~\ref{fig:alg}, which induces the supervisor \(S_2\) whose controlled behavior is as shown in Figure~\ref{fig:G2}. Alternatively, if we choose \(\Sigma\setminus\{a\}\), procedure \texttt{Extract} expands the control structure until it revisits the state \(\{ (5,(\texttt{N}, \texttt{Y}, \texttt{N})), (6,(\texttt{N}, \texttt{N}, \texttt{N})) \}\). This yields a different structure shown in Figure~\ref{fig:controlstructure}, corresponding to the supervisor \(S_1\), whose controlled behavior is as shown in Figure~\ref{fig:G1}.
\end{myexm}

\subsection{Correctness, Permissiveness and Complexity Analysis}
In this subsection, we analyze the correctness of the supervisory synthesis algorithm. 
First, we establish the soundness of the algorithms, which states that the synthesized supervisor is indeed live and enforces the prediction-based property.

\begin{mylem}\upshape
Given the structure \(\S\) returned by Algorithm~\ref{alg:main}, the partial-observation supervisor \(S: \Sigma_o^* \to \Gamma\) induced from \(\S\) provides a solution to Problem~\ref{problem}.
\end{mylem}
\begin{proof}
The proof follows directly from Theorem~\ref{thm1} and the construction of \(\B\), which ensures that all decision-states in \(\B\) are live and safe. Since 
$\S$ is extracted as a subgraph by construction, all decision-states in \(\S\) are also live and safe. Therefore, \(\S\) constitutes a solution to Problem~\ref{problem}.
\end{proof}

Note that our synthesis algorithm restricts the solution space of the supervisor to the IS-based control structure. In general, a supervisor may require more memory than what can be represented as an IS-based control structure. However, the following result shows that this restriction is without loss of generality, thereby establishing the completeness of the algorithm.

\begin{mylem}\label{lem:complete}\upshape
Algorithm~\ref{alg:main} will not return ``no solution exists" when a solution to Problem~\ref{problem} exists.
\end{mylem}
\begin{proof}
The proof is provided in the Appendix.
\end{proof}

By combining Lemmas~1 and~2, we can finally establish the correctness of the synthesis algorithm.
\begin{mythm}\upshape
    Algorithm~\ref{alg:main} correctly solves Problem~\ref{problem}, i.e., it is both sound and complete.
\end{mythm}

We finally show that the supervisor synthesized is maximally permissive in the sense of language inclusion.

\begin{mythm}\label{thm:perm}\upshape
Given the structure \(\S\) returned by Algorithm~\ref{alg:main}, the partial-observation supervisor \(S: \Sigma_o^* \to \Gamma\) induced from \(\S\) is maximally permissive, i.e.,
\[
(\forall S' \text{ that solves Probelm~\ref{problem}})[\mathcal{L}(S/G)\not\subset\mathcal{L}(S'/G)]
\]
\end{mythm}
\begin{proof}
The proof is provided in the Appendix.
\end{proof}

We conclude this section by discussing the complexity of the control synthesis algorithm. To synthesize an IS-based supervisor, we first need to construct the structure \(\B\), which contains at most \((1+3^{|H|})^{|X|}\) decision-states. For each decision-state, there are at most \(2^{|\Sigma_c|}\) control decisions defined,  for each decision, there are at most $|\mathbb{I}^\B|^{1+|\Sigma_o|}=(1+3^{|H|})^{|X|\cdot(1+|\Sigma_o|)}$ observation-states defined, and thus at most $(1+3^{|H|})^{|X|\cdot(1+|\Sigma_o|)}$ transitions defined. 
Finally, for each observation-state, there are at most $|\Sigma_o|$ transitions defined. 
Therefore, in the worst case, the largest possible \(\B\) contains \((1+3^{|H|})^{|X|}+2^{|\Sigma_c|}\cdot (1+3^{|H|})^{|X|\cdot(|\Sigma_o|+2)} \) states and \((1+|\Sigma_o|)\cdot 2^{|\Sigma_c|}\cdot (1+3^{|H|})^{|X|\cdot(|\Sigma_o|+2)}\) transitions. The complexity of  procedure \texttt{Prune} is quadratic in the size of \(\B\). The complexity of  procedure \texttt{Extract} is linear in the size of the pruned \(\B\) which also has \((1+3^{|H|})^{|X|}+2^{|\Sigma_c|}\cdot (1+3^{|H|})^{|X|\cdot(|\Sigma_o|+2)} \) states in the worst case. Therefore, the entire complexity of the proposed control synthesis algorithm is exponential in the size of the system \(G\) and the horizon of the prediction vector \(\mathbb{V}\). However, since synthesizing a partial observation supervisor is inherently PSPACE-hard, this complexity seems unavoidable due to the partial observation nature of our problem.

\section{Conclusion}\label{sec: conc}
In this paper, we formulated the general notion of prediction-based properties and solved the corresponding supervisory synthesis problem.
We demonstrated that the proposed framework is quite general and can model many practical problems, such as active fault prediction and intention-security protection in partially-observed DES.
Our results also introduced a novel information structure that effectively addresses the fundamental challenge in synthesizing supervisors for properties related to undetermined future behaviors.
Note that, in this work, the prediction horizon of interest must be a given finite horizon. 
As future work, we aim to extend our framework to observational properties involving infinite prediction horizons.

\appendix
\textbf{Proof of Proposition~\ref{prop:IS}} 
\begin{proof}
    We first prove $\state(\mathbb{I}^\S(\alpha))=\{\delta(s)\in X: s\in \mathcal{O}(\alpha)\}$ by induction on the length of $\alpha$. For simplicity, we denote $\mathcal{E}_\mathcal{O}(\alpha)=\{\delta(s)\in X: s\in \mathcal{O}(\alpha)\}$.

    \emph{Induction Basis}:
    Suppose that $|\alpha|=0$. Then we know that $\mathbb{I}^\S(\epsilon)=\imath_0^\S=\{(x_0,\mathbf{v})\}$. It is clear that $\{x_0\}=\mathcal{E}_\mathcal{O}(\epsilon)$. Therefore, the induction basis holds.

    \emph{Induction Step}:
    Now, suppose that $\state(\mathbb{I}^\S(\alpha))=\mathcal{E}_\mathcal{O}(\alpha)$ holds for $|\alpha|=k$, then we prove that it also holds for $\alpha\sigma\in P(\mathcal{L}(S/G))$, where $\sigma\in\Sigma_o$.
    According to Definition~\ref{def:sig-obs-inf-patt}, we have that 
    \begin{align}
        \state&(\mathbb{I}^\S(\alpha\sigma)) \nonumber \\
        &=\OR_\sigma(\UR_{S(\alpha)}(\state(\mathbb{I}^\S(\alpha)) \nonumber \\
        &=\OR_\sigma(\UR_{S(\alpha)}(\mathcal{E}_\mathcal{O}(\alpha)))  \nonumber \\
        &=\OR_\sigma\left(\left\{\delta(x,w)\in X: \!\!\!\!\!\!
        \begin{array}{cc}
             & x\in \mathcal{E}_\mathcal{O}(\alpha), \\
             & w\in(\Sigma_{uo}\cap S(\alpha))^*
        \end{array}\!\!\!\!
        \right\}\right) \nonumber \\
        &=\OR_\sigma(\mathcal{E}_S(\alpha)) \nonumber \\
        &=\{\delta(s\sigma)\in X: s\in P^{-1}_S(\alpha)\} \nonumber \\
        &=\{\delta(s\sigma)\in X: s\sigma\in \mathcal{O}(\alpha\sigma)\} = \mathcal{E}_\mathcal{O}(\alpha\sigma)
    \end{align}
    This completes the induction step.

    Then we prove that $\mathsf{vec}(\mathbb{I}^\S(\alpha))=\Xi^S(\alpha)$ by proving the following more strict claim:
    For each observation-state $\mathbb{A}^\S(\alpha)=(\imath_u,\{\imath_\sigma\}_{\sigma\in \Sigma_o\cap S(\alpha)})$ induced by $\alpha\in P(\L(S/G))$, we have $\mathbf{v}^{\delta(s)}_{\imath_u}=\xi^S(s)$ for each prediction vector, where $s\in P^{-1}_S(\alpha)$. Since $\mathbb{I}^\S(\alpha)\subseteq \imath_u$ and $\textsf{state}(\mathbb{I}^\S(\alpha))=\mathcal{E}_\mathcal{O}(\alpha)$, the above claim immediately leads to our conclusion.
    
    We prove this by induction on the dimension of the prediction vector $\mathbf{v}^{\delta(s)}_{\imath_u}$. 

    \emph{Induction Basis}:
    We first prove that 
    $\mathbf{v}^{\delta(s)}_{\imath_u}[0]=\xi^S(s)[0]$ holds for all $s\in P^{-1}_S(\alpha),  \alpha\in P(\mathcal{L}(S/G))$. According to Definition~\ref{def:inf-consis}, every augmented state $\tilde{x}=(x,\mathbf{v})\in \imath_u$ is consistent for the current instant, i.e., $\mathbf{v}[0]=\texttt{Y}$ if $x\in X_C$ and $\mathbf{v}[0]=\texttt{N}$ if $x\not\in X_C$. 
    Therefore, we have $\mathbf{v}^{\delta(s)}_{\imath_u}[0]=\chi_C(\{\delta(s)\})$ according to the definition of $\chi_C$.
    On the other hand, we also have $\xi^S(s)[0]=\chi_C(\{\delta(s)\})$ according to Definition~\ref{def:pre-vec}. Moreover, we have $\textsf{state}(\imath_u)=\mathcal{E}_S(\alpha)$ according to Definition~\ref{def:sig-obs-inf-patt} and Equation~(29), i.e., $(\delta(s),\mathbf{v}^{\delta(s)}_{\imath_u})\in\imath_u$ for all $s\in P^{-1}_S(\alpha)$. Therefore, we can conclude that the induction basis holds.
    
    \emph{Induction Step}: Now we suppose that $\mathbf{v}^{\delta(s)}_{\imath_u}[k]=\xi^S(s)[k]$ holds for all $s\in P^{-1}_S(\alpha),  \alpha\in P(\L(S/G))$ for some $0\leq k\leq H-1$, then we prove that it also holds for the $k+1$-th dimension. According to Definition~\ref{def:inf-consis}, every augmented state $\tilde{x}=(x,\mathbf{v})\in \imath_u$ is consistent for the future instant $k+1$, i.e.,
    \begin{equation}
  \mathbf{v}[k+1]=  		
  \left\{
		\begin{array}{ll}
			 \texttt{Y}  & \text{if}\quad \forall (x',\mathbf{v'})\in \mathcal{R}_{\mathbb{A}^\S(\alpha)}(\tilde{x}): \mathbf{v'}[k]=\texttt{Y}  \\
			 \texttt{N}  & \text{if}\quad \forall (x',\mathbf{v'})\in \mathcal{R}_{\mathbb{A}^\S(\alpha)}(\tilde{x}): \mathbf{v'}[k]=\texttt{N}\\
			 \texttt{U}  & \text{otherwise} 
		\end{array}
		\right.  \nonumber
\end{equation}    

According to Definition~\ref{def:one-step-rea-set}, we know that for an augmented state $(\delta(s),\mathbf{v})\in\imath_u$, we have $\mathcal{R}_{\mathbb{A}^\S(\alpha)}((\delta(s),\mathbf{v}))=\{(\delta(s\sigma),\mathbf{v}'):\sigma\in\gamma\}$. Since for dimension $k$ we have $\mathbf{v}^{\delta(s\sigma)}_{\imath_u}[k]=\xi^S(s\sigma)[k]$, we can conclude that 

    \begin{equation}\label{eq:30}
  \mathbf{v}^{\delta(s)}_{\imath_u}[k+1]=  		
  \left\{
		\begin{array}{ll}
			 \texttt{Y}  & \text{if}\quad \forall \sigma\in\gamma: \xi^S(s\sigma)[k]=\texttt{Y}  \\
			 \texttt{N}  & \text{if}\quad \forall \sigma\in\gamma: \xi^S(s\sigma)[k]=\texttt{N}\\
			 \texttt{U}  & \text{otherwise} 
		\end{array}
		\right. 
\end{equation}    
On the other hand, we have $\xi^S(s)[k+1]=\chi_C(\text{Reach}_{k+1}(s))$, whose value depends on the relationship between $\text{Reach}_{k+1}(s)$ and $X_C$. Note that we have $\text{Reach}_{k+1}(s)=\text{Reach}_{k}(\text{Reach}_{1}(s))$ according to Equation~(\ref{eq:k-reachable-set}). We take the condition $\xi^S(s)[k+1]=\texttt{Y}$ as an example, which holds if  $\text{Reach}_{k+1}(s)\subseteq X_C$ according to the definition of $\chi_C$, and thus $\text{Reach}_{k}(\text{Reach}_{1}(s))\subseteq X_C$, which is equivalent to $\text{Reach}_{k}(s\sigma)\subseteq X_C$ for all $\sigma\in \gamma=S(P(s))$, i.e., $\xi^S(s\sigma)[k]=\texttt{Y}$ for all $\sigma\in \gamma$. According to Equation~(\ref{eq:30}), we can conclude that $\mathbf{v}^{\delta(s)}_{\imath_u}[k+1]=\xi^S(s)[k+1]$, which completes our induction step. 
\end{proof}

\textbf{Proof of Theorem~\ref{thm1}}
\begin{proof}
We first note that, by Proposition~\ref{prop:IS} and Definition~\ref{def:con-struc}, given an IS-based supervisor $S$, for any string $\alpha\in P(\mathcal{L}(S/G))$, the decision state $\mathbb{I}^\S(\alpha)$ reached is sufficient to check the value of $\Phi(\Xi^S(\alpha))$ since we have $\mathsf{vec}(\mathbb{I}^\S(\alpha))=\Xi^S(\alpha)$. Therefore, if all decision states in $\S$ is safe, then we have $\Phi(\Xi^S(\alpha))=1$ for all $\alpha\in P(\L(S/G))$, i.e., $S/G\models(X_C,\Phi)$.

Then we prove that the live decision state indeed captures the property of liveness. Mathematically, we have
\[
\forall \alpha\in P(\mathcal{L}(S/G)): \mathbb{I}^\S(\alpha) \text{ is live}\Leftrightarrow \mathcal{L}(S/G) \text{ is live}.
\]
By contrapositive, we know that $\L(S/G)$ is not live if there exist $s\in \mathcal{L}(S/G)$, such that for any $\sigma\in S(P(s))$, $\delta(s\sigma)$ is not defined. 
According to the proof of Proposition~\ref{prop:IS}, by taking the same string $s$, we have that the state $\delta(s)\in \state(\imath_u)$ has no successor state $\delta(s\sigma)$ defined, where $\mathbb{A}^\S(P(s))=(\imath_u,\{\imath_\sigma\})$. According to Definition~\ref{def:sig-obs-inf-patt}, we have that $\exists x\in \UR_\gamma(  \state(\mathbb{I}^\S(P(s))) ),\forall \sigma\in \gamma: \delta(x,\sigma)$ is not defined, i.e., $\mathbb{I}^\S(P(s))$ is not live. For another side, assume there exists a string $s\in \mathcal{L}(S/G)$ such that $\mathbb{I}^\S(P(s))$ is not live under \(\gamma=S(P(s))\), then there exist a string $s'\in \mathcal{L}(S/G)$ such that $P(s)=P(s')$ and $\delta(\delta(s'),\sigma)$ is not defined for any $\sigma\in S(P(s'))=S(P(s))$. We can also conclude that $S/G$ is not live by taking the same $s'$.
\end{proof}

\textbf{Proof of Lemma~\ref{lem:complete}}
\begin{proof}
We prove by showing that \(\mathbb{I}^{\B}_0\ne\emptyset\) in line~7 when a non-IS-based supervisor that solves Problem~\ref{problem} exists.

Assume there exists a language-based supervisor \(S:P(\mathcal{L}(G))\to\Gamma\) that solves Problem~\ref{problem}. We first construct a decision-observation structure \(B^S\) as follows:

We start by initializing a finite decision-state space
\[
\mathbb{I}^S=\big\{\{(\delta(s),\xi^S(s)):s\in\O(\alpha)\} \in\mathbb{I}: \alpha\in P(\L(S/G))\big\}.
\]
Then for each \(\imath\in\mathbb{I}^\S\), for each \(\alpha\in P(\L(S/G))\) such that \(\{(\delta(s),\xi^S(s)):s\in\O(\alpha)\}=\imath\), we define a transition \((\imath,S(\alpha),\mathcal{I})\) from decision-state \(\imath\) to observation-state \(\mathcal{I}\), where 
\begin{align}
    \mathcal{I}=\left(\begin{array}{cc}
         & \{(\delta(s),\xi^S(s)):s\in P^{-1}(\alpha)\},\\
         & \{\{(\delta(s),\xi^S(s)):s\in\O(\alpha\sigma)\}\}_{\sigma\in\Sigma_o\cap S(\alpha)}
    \end{array}
    \right)\nonumber
\end{align}
And we define transitions \((\mathcal{I},\sigma,\imath')\) for each \(\sigma\in\Sigma_o\cap S(\alpha)\), \(\imath'= \{(\delta(s),\xi^S(s)):s\in\O(\alpha\sigma)\}\).
Finally, we set \(\imath_0=\{(x_0,\xi^S(\epsilon))\}\) as the initial decision state.

By the above construction, for each observation sequence \(\alpha=\sigma_1\sigma_2...\sigma_n\in P(\mathcal{L}(S/G))\), it also induces an unique path in \(B^S\)
\begin{equation}
    \imath_0\xrightarrow{S(\epsilon)}\mathcal{I}_0\xrightarrow{\sigma_1}\imath_1\xrightarrow{S(\sigma_1)}\cdots \xrightarrow{\sigma_n}\imath_n \xrightarrow{S(\sigma_1...\sigma_{n})}\mathcal{I}_{n},\nonumber
\end{equation}
such that \(\mathsf{state}(\imath_n)=\{\delta(x)\in X:s\in\O(\alpha)\}\) and \(\mathsf{vec}(\imath_n)=\Xi^S(\alpha)\). 

Since \(S\) solves Problem~1, we can conclude that each decision-state \(\imath\) in \(B^S\) is safe, and is also live under each decision \(\gamma\) such that \((\imath,\gamma,\mathcal{I})\) is defined on \(\imath\). Also, according to the definition of \(\chi_C\) and \(\text{Reach}_k\), for each transition \((\imath,\gamma,\mathcal{I})\) defined in \(B^S\), we have \(\mathcal{I}\in\mathbb{A}(\imath,\gamma)\). Therefore, we have \(B^S\sqsubseteq B\) according to procedure \texttt{Expand}.

Next, we prove that \(B^S\sqsubseteq \texttt{Prune}(B)\), i.e., \(B\) is at least as large as \(B^S\) after calling procedure \texttt{Prune}. This conclusion directly follows the above construction, where each decision-state \(\imath\) in \(B^S\) has at least one successor observation-state, and each observation-state \(\mathcal{I}\) in \(B^S\) has all transitions defined within the corresponding feasible observation events, and thus are complete. Since we have \(B^S\sqsubseteq B\), all states that are complete in \(B^S\) are also complete in \(B\), and thus will not be removed in procedure \texttt{Prune}.

Finally, we can conclude that the initial state \(\imath_0\) of \(B^S\) is also an initial state of \(B\) after procedure \texttt{Prune}, i.e., \(\mathbb{I}^{\B}_0\ne\emptyset\), and Algorithm~\ref{alg:main} will not return "no solution exists" in this case.
\end{proof}

\textbf{Proof of Theorem~\ref{thm:perm}}
\begin{proof}
We prove by contradiction. Assume that there exists supervisor $S'$ such that $S'$ solves Problem~\ref{problem} and $\mathcal{L}(S /G) \subset \mathcal{L}(S'/G)$.  Without loss of generality, we assume that the supervisors $S$ and $S'$ are both \emph{irredundant} such that each enabled controllable event in the control decisions is useful. Since $\mathcal{L}(S /G) \subset \mathcal{L}(S'/G)$, we have
\[
\forall s \in \mathcal{L}(S/G) \subset \mathcal{L}(S'/G), \quad S(P(s)) \subseteq S'(P(s)).
\]
Then we can conclude that for any $s \in \mathcal{L}(S/G) \subset \mathcal{L}(S'/G)$ and $s' \in P^{-1}_{S}(P(s)) \subseteq P^{-1}_{S'}(P(s))$, we have
\[
\text{Reach}^{S}_{k}(s') \subseteq \text{Reach}^{S'}_{k}(s').
\]
According to the definition of the \(\chi_C\), we know that for any such $s'$, \(\xi^S(s')[i]=\texttt{Y}\implies \xi^{S'}(s')[i]\ne\texttt{N}\), \(\xi^S(s')[i]=\texttt{N}\implies \xi^{S'}(s')[i]\ne\texttt{Y}\) and \(\xi^S(s')[i]=\texttt{U}\implies \xi^{S'}(s')[i]=\texttt{U}\).

Moreover, since $\mathcal{L}(S /G) \subset \mathcal{L}(S'/G)$, there exist some observation sequence $\alpha=\sigma_1\sigma_2...\sigma_n \in P(\mathcal{L}(S/G))$, such that $S(\alpha) \subset S'(\alpha)$ and for any $\alpha'\in\overline{\{\alpha\}}\setminus\{\alpha\}$, we have $S(\alpha')=S'(\alpha')$. We now consider the structure \(\S\) and \(B^{S'}\), both of which are included in \(B\) after calling \texttt{Prune} according to Algorithm~\ref{alg:main}and the proof of lemma~2. Consider two paths
\begin{equation}
    \imath_0\xrightarrow{S(\epsilon)}\mathcal{I}_0\xrightarrow{\sigma_1}\imath_1\xrightarrow{S(\sigma_1)}\cdots \xrightarrow{\sigma_n}\imath_n \xrightarrow{S(\sigma_1...\sigma_{n})}\mathcal{I}_{n},\nonumber
\end{equation}
and
\begin{equation}
    \imath'_0\xrightarrow{S'(\epsilon)}\mathcal{I}'_0\xrightarrow{\sigma_1}\imath'_1\xrightarrow{S'(\sigma_1)}\cdots \xrightarrow{\sigma_n}\imath_n \xrightarrow{S'(\sigma_1...\sigma_{n})}\mathcal{I}_{n}\nonumber
\end{equation}
led by \(S\) and \(S'\), where we have \(\mathsf{state}(\imath_i)=\mathsf{state}(\imath'_i)\) for all \(i\in\{0,1,..,n\}\). Since \(\mathsf{vec}(\imath_i)=\Xi^S(\sigma_1..\sigma_i)\) and \(\mathsf{vec}(\imath'_i)=\Xi^{S'}(\sigma_1..\sigma_i)\) by our construction, we have \(\imath_i\leq_\texttt{U}\imath'_i\) for all \(i\in\{0,1,..,n\}\).

Starting from \(\imath_0\) and \(\imath'_0\), according to line~10 in the algorithm, we have \(\imath_0\not<_\texttt{U}\imath'_0\), also since \(\text{Reach}^{S}_{k}(s) \subseteq \text{Reach}^{S'}_{k}(s)\) for all \(s\in\L(S/G)\), we can conclude that \(\imath_0=\imath'_0\).
Then, according to line~2 in procedure \texttt{Extract}, we have \(\mathcal{I}_0\not<_\texttt{U}\mathcal{I}'_0\), and thus \(\imath_1\not<_\texttt{U}\imath'_1\), then we can also conclude that \(\imath_1=\imath'_1\).
Therefore, we eventually have \(\imath_n=\imath'_n\) due to the same reason, where we have \(S(\alpha)\subset S'(\alpha)\) by our assumption. However, this leads to a conflict with line~1 of procedure \texttt{Extract}, which states that \(S(\alpha)\not\subset S'(\alpha)\).
\end{proof}

\bibliographystyle{plain}
\bibliography{references}

\begin{thebibliography}{10}

\bibitem{balun2021comparing}
Ji{\v{r}}{\'\i} Balun and Tom{\'a}{\v{s}} Masopust.
\newblock Comparing the notions of opacity for discrete-event systems.
\newblock {\em Discrete Event Dynamic Systems}, 31(4):553--582, 2021.

\bibitem{barcelos2021enforcing}
Raphael~Julio Barcelos and Jo{\~a}o~Carlos Basilio.
\newblock Enforcing current-state opacity through shuffle and deletions of
  event observations.
\newblock {\em Automatica}, 133:109836, 2021.

\bibitem{barrett2000separation}
George Barrett and St{\'e}phane Lafortune.
\newblock On the separation of estimation and control in discrete-event
  systems.
\newblock In {\em 39th IEEE Conference on Decision and Control}, volume~3,
  pages 2258--2259. IEEE, 2000.

\bibitem{cao2024active}
Lin Cao, Shaolong Shu, and Feng Lin.
\newblock Active fault isolation for discrete event systems.
\newblock {\em IEEE Transactions on Automatic Control}, 69(8):4988--5003, 2024.

\bibitem{cassandras2008introduction}
Christos~G Cassandras and St{\'e}phane Lafortune.
\newblock {\em Introduction to discrete event systems}.
\newblock Springer, 2008.

\bibitem{chen2014stochastic}
Jun Chen and Ratnesh Kumar.
\newblock Stochastic failure prognosability of discrete event systems.
\newblock {\em IEEE Transactions on Automatic Control}, 60(6):1570--1581, 2014.

\bibitem{chen2023you}
Yu~Chen, Shuo Yang, Rahul Mangharam, and Xiang Yin.
\newblock You don't know when i will arrive: Unpredictable controller synthesis
  for temporal logic tasks.
\newblock {\em IFAC-PapersOnLine}, 56(2):3591--3597, 2023.

\bibitem{cieslak1988supervisory}
Randy Cieslak, C~Desclaux, Ayman~S Fawaz, and Pravin Varaiya.
\newblock Supervisory control of discrete-event processes with partial
  observations.
\newblock {\em IEEE transactions on automatic control}, 33(3):249--260, 1988.

\bibitem{dubreil2010supervisory}
J{\'e}r{\'e}my Dubreil, Philippe Darondeau, and Herv{\'e} Marchand.
\newblock Supervisory control for opacity.
\newblock {\em IEEE Transactions on Automatic Control}, 55(5):1089--1100, 2010.

\bibitem{genc2009predictability}
Sahika Genc and St{\'e}phane Lafortune.
\newblock Predictability of event occurrences in partially-observed
  discrete-event systems.
\newblock {\em Automatica}, 45(2):301--311, 2009.

\bibitem{haar2020active}
Stefan Haar, Serge Haddad, Stefan Schwoon, and Lina Ye.
\newblock Active prediction for discrete event systems.
\newblock In {\em 40th IARCS Annual Conference on Foundations of Software
  Technology and Theoretical Computer Science}, 2020.

\bibitem{hadjicostis2020estimation}
Christoforos~N Hadjicostis.
\newblock {\em Estimation and inference in discrete event systems}.
\newblock Springer, 2020.

\bibitem{hou2022abstraction}
Junyao Hou, Siyuan Liu, Xiang Yin, and Majid Zamani.
\newblock Abstraction-based verification of approximate preopacity for control
  systems.
\newblock {\em IEEE Control Systems Letters}, 7:1087--1092, 2022.

\bibitem{hu2020design}
Yihui Hu, Ziyue Ma, and Zhiwu Li.
\newblock Design of supervisors for active diagnosis in discrete event systems.
\newblock {\em IEEE Transactions on Automatic Control}, 65(12):5159--5172,
  2020.

\bibitem{hu2021diagnosability}
Yihui Hu, Ziyue Ma, Zhiwu Li, and Alessandro Giua.
\newblock Diagnosability enforcement in labeled petri nets using supervisory
  control.
\newblock {\em Automatica}, 131:109776, 2021.

\bibitem{kumar2015stochastic}
Panqanamala~Ramana Kumar and Pravin Varaiya.
\newblock {\em Stochastic systems: Estimation, identification, and adaptive
  control}.
\newblock SIAM, 2015.

\bibitem{lin2011opacity}
Feng Lin.
\newblock Opacity of discrete event systems and its applications.
\newblock {\em Automatica}, 47(3):496--503, 2011.

\bibitem{lin1988observability}
Feng Lin and Walter~Murray Wonham.
\newblock On observability of discrete-event systems.
\newblock {\em Information sciences}, 44(3):173--198, 1988.

\bibitem{liu2022enforcement}
Rongjian Liu and Jianquan Lu.
\newblock Enforcement for infinite-step opacity and k-step opacity via
  insertion mechanism.
\newblock {\em Automatica}, 140:110212, 2022.

\bibitem{moulton2022using}
Richard~Hugh Moulton, Behnam~Behinaein Hamgini, Zahra~Abedi Khouzani,
  R{\^o}mulo Meira-G{\'o}es, Fei Wang, and Karen Rudie.
\newblock Using subobservers to synthesize opacity-enforcing supervisors.
\newblock {\em Discrete Event Dynamic Systems}, 32(4):611--640, 2022.

\bibitem{ran2022prognosability}
Ning Ran, Jinyuan Hao, and Carla Seatzu.
\newblock Prognosability analysis and enforcement of bounded labeled petri
  nets.
\newblock {\em IEEE Transactions on Automatic Control}, 67(10):5541--5547,
  2022.

\bibitem{reijnen2022supervisory}
Ferdie~FH Reijnen, Toby~R Erens, Joanna~M van~de Mortel-Fronczak, and Jacobus~E
  Rooda.
\newblock Supervisory controller synthesis and implementation for safety plcs.
\newblock {\em Discrete Event Dynamic Systems}, 32(1):115--141, 2022.

\bibitem{reijnen2020modeling}
Ferdie~FH Reijnen, Martijn~A Goorden, Joanna~M van~de Mortel-Fronczak, and
  Jacobus~E Rooda.
\newblock Modeling for supervisor synthesis--a lock-bridge combination case
  study.
\newblock {\em Discrete Event Dynamic Systems}, 30:499--532, 2020.

\bibitem{ritsuka2025joint}
K~Ritsuka, St{\'e}phane Lafortune, and Feng Lin.
\newblock Joint opacity and opacity against state-estimate-intersection-based
  intrusion of discrete-event systems.
\newblock {\em Automatica}, 176:112136, 2025.

\bibitem{rosa2024modular}
Marcelo Rosa, Jos{\'e}~ER Cury, and Fabio~L Baldissera.
\newblock A modular synthesis approach for the coordination of multi-agent
  systems: the multiple team case.
\newblock {\em Discrete Event Dynamic Systems}, 34(1):163--198, 2024.

\bibitem{sampath1998active}
Meera Sampath, St{\'e}phane Lafortune, and Demosthenis Teneketzis.
\newblock Active diagnosis of discrete-event systems.
\newblock {\em IEEE transactions on automatic control}, 43(7):908--929, 1998.

\bibitem{shu2012delayed}
Shaolong Shu and Feng Lin.
\newblock Delayed detectability of discrete event systems.
\newblock {\em IEEE Transactions on Automatic Control}, 58(4):862--875, 2012.

\bibitem{shu2013enforcing}
Shaolong Shu and Feng Lin.
\newblock Enforcing detectability in controlled discrete event systems.
\newblock {\em IEEE Transactions on Automatic Control}, 58(8):2125--2130, 2013.

\bibitem{shu2007detectability}
Shaolong Shu, Feng Lin, and Hao Ying.
\newblock Detectability of discrete event systems.
\newblock {\em IEEE Transactions on Automatic Control}, 52(12):2356--2359,
  2007.

\bibitem{sweeney2002k}
Latanya Sweeney.
\newblock $k$-anonymity: A model for protecting privacy.
\newblock {\em International Journal of Uncertainty, Fuzziness and
  Knowledge-Based Systems}, 10(05):557--570, 2002.

\bibitem{takai2015robust}
Shigemasa Takai.
\newblock Robust prognosability for a set of partially observed discrete event
  systems.
\newblock {\em Automatica}, 51:123--130, 2015.

\bibitem{thuijsman2024supervisory}
Sander Thuijsman and Michel Reniers.
\newblock Supervisory control for dynamic feature configuration in product
  lines.
\newblock {\em ACM Transactions on Embedded Computing Systems}, 23(5):1--25,
  2024.

\bibitem{tong2018current}
Yin Tong, Zhiwu Li, Carla Seatzu, and Alessandro Giua.
\newblock Current-state opacity enforcement in discrete event systems under
  incomparable observations.
\newblock {\em Discrete Event Dynamic Systems}, 28:161--182, 2018.

\bibitem{watanabe2021fault}
Ana~TY Watanabe, Renan Sebem, Andre~B Leal, and Marcelo da~S Hounsell.
\newblock Fault prognosis of discrete event systems: An overview.
\newblock {\em Annual Reviews in Control}, 51:100--110, 2021.

\bibitem{wintenberg2021enforcement}
Andrew Wintenberg, Matthew Blischke, St{\'e}phane Lafortune, and Necmiye Ozay.
\newblock Enforcement of k-step opacity with edit functions.
\newblock In {\em 60th IEEE Conference on Decision and Control (CDC)}, pages
  331--338. IEEE, 2021.

\bibitem{wintenberg2022general}
Andrew Wintenberg, Matthew Blischke, St{\'e}phane Lafortune, and Necmiye Ozay.
\newblock A general language-based framework for specifying and verifying
  notions of opacity.
\newblock {\em Discrete Event Dynamic Systems}, 32(2):253--289, 2022.

\bibitem{wonham2019supervisory}
W~Murray Wonham and Kai Cai.
\newblock Supervisory control of discrete-event systems, 2019.

\bibitem{xie2024optimal}
Yifan Xie, Shaoyuan Li, and Xiang Yin.
\newblock Optimal synthesis of opacity-enforcing supervisors for qualitative
  and quantitative specifications.
\newblock {\em IEEE Transactions on Automatic Control}, 2024.

\bibitem{yang2022secure}
Shuo Yang and Xiang Yin.
\newblock Secure your intention: On notions of pre-opacity in discrete-event
  systems.
\newblock {\em IEEE Transactions on Automatic Control}, 68(8):4754--4766, 2023.

\bibitem{yin2016synthesis}
Xiang Yin and St{\'e}phane Lafortune.
\newblock Synthesis of maximally permissive supervisors for partially-observed
  discrete-event systems.
\newblock {\em IEEE Transactions on Automatic Control}, 61(5):1239--1254, 2016.

\bibitem{yin2016uniform}
Xiang Yin and St{\'e}phane Lafortune.
\newblock A uniform approach for synthesizing property-enforcing supervisors
  for partially-observed discrete-event systems.
\newblock {\em IEEE Transactions on Automatic Control}, 61(8):2140--2154, 2016.

\bibitem{yin2019supervisory}
Xiang Yin and Shaoyuan Li.
\newblock Supervisory control for delayed detectability of discrete event
  systems.
\newblock In {\em IEEE 15th International Conference on Automation Science and
  Engineering (CASE)}, pages 480--485. IEEE, 2019.

\bibitem{yin2020synthesis}
Xiang Yin and Shaoyuan Li.
\newblock Synthesis of dynamic masks for infinite-step opacity.
\newblock {\em IEEE Transactions on Automatic Control}, 65(4):1429--1441, 2020.

\bibitem{yin2016decentralized}
Xiang Yin and Zhaojian Li.
\newblock Decentralized fault prognosis of discrete event systems with
  guaranteed performance bound.
\newblock {\em Automatica}, 69:375--379, 2016.

\bibitem{you2019verification}
Dan You, ShouGuang Wang, and Carla Seatzu.
\newblock Verification of fault-predictability in labeled petri nets using
  predictor graphs.
\newblock {\em IEEE Transactions on Automatic Control}, 64(10):4353--4360,
  2019.

\bibitem{zhang2023polynomial}
Kuize Zhang.
\newblock Polynomial-time verification and enforcement of delayed strong
  detectability for discrete-event systems.
\newblock {\em IEEE Transactions on Automatic Control}, 68(1):510--515, 2023.

\end{thebibliography}

\end{document}